\documentclass[sigconf]{aamas} 

\usepackage{balance} 

\usepackage{framed} 
\usepackage{tikz}
\usepackage{pifont}
\usepackage{subfigure}
\usepackage{colortbl}
\usepackage{natbib}

\setcopyright{ifaamas}
\acmConference[AAMAS '25]{Proc.\@ of the 24th International Conference
on Autonomous Agents and Multiagent Systems (AAMAS 2025)}{May 19 -- 23, 2025}
{Detroit, Michigan, USA}{A.~El~Fallah~Seghrouchni, Y.~Vorobeychik, S.~Das, A.~Nowe (eds.)}
\copyrightyear{2025}
\acmYear{2025}
\acmDOI{}
\acmPrice{}
\acmISBN{}

\title[]{Housing Market on Networks}

\author{Xinwei Song}
\affiliation{
  \institution{Key Laboratory of Intelligent Perception and Human-Machine Collaboration, ShanghaiTech University}
  \city{Shanghai}
  \country{China}}
\email{songxw2024@shanghaitech.edu.cn}

\author{Tianyi Yang}
\affiliation{
  \institution{Key Laboratory of Intelligent Perception and Human-Machine Collaboration, ShanghaiTech University}
  \city{Shanghai}
  \country{China}}
\email{yty2@shanghaitech.edu.cn}

\author{Dengji Zhao}
\affiliation{
  \institution{Key Laboratory of Intelligent Perception and Human-Machine Collaboration, ShanghaiTech University}
  \city{Shanghai}
  \country{China}}
\email{zhaodj@shanghaitech.edu.cn}

\begin{abstract}
Incentivizing the existing participants to invite new participants to join an auction, matching or cooperative game have been extensively studied recently. One common challenge to design such incentive in these games is that the invitees and inviters are competitors. To have such an incentive, we normally have to sacrifice some of the traditional properties. Especially, in a housing market (one kind of one-sided matching), we cannot maintain the traditional stability and optimality. The previous studies proposed some new matching mechanisms to have the invitation incentive (part of the incentive compatibility), but did not have any guarantee on stability and optimality.

In this paper, we propose new notions of stability and optimality which are achievable with incentive compatibility. We weaken stability and optimality on a special structure (complete components) on networks. We first prove that the weakened notions are the best we can achieve with incentive compatibility.
Then, we propose three mechanisms (Swap With Neighbors, Leave and Share, and Connected Trading Cycles) to satisfy the desirable properties. Connected Trading Cycles is the first mechanism to satisfy the best stability and optimality compatible with incentive compatibility.
\end{abstract}


\newcommand{\BibTeX}{\rm B\kern-.05em{\sc i\kern-.025em b}\kern-.08em\TeX}

\begin{document}


\pagestyle{fancy}
\fancyhead{}


\maketitle 

\section{Introduction}
Housing market on networks studies an endowment exchange problem with special consideration of the participants' social connections~\cite{DBLP:conf/atal/KawasakiWTY21}. In this model, participants' social connections are private information and only a small group initiates the matching. Those who are already in the game can decide whether or not to invite their neighbors. Expanding the market via invitation can be beneficial because more participants will provide more matching options, which opens the possibility of a more satisfactory matching. Several real-world applications, like online housing markets\footnote{https://www.homeexchange.com} and second-hand goods exchange platforms\footnote{https://www.freecycle.org}, have utilized social networks to promote their markets. On these platforms, existing participants are encouraged to share the matching information on their social media and invite their friends to form larger markets. But sometimes participants may hesitate to invite because their friends will compete with them and take away their favorite items in the matching game. For this reason, the main obstacle is ensuring that the inviters' match is not getting worse after inviting others. 

To solve the above problem, one way is to constrain participants' selection space to remove potential competition caused by invitation. For example, \citet{DBLP:conf/atal/KawasakiWTY21} restricted the social network between agents as trees and proposed Modified TTC which only allows participants to choose from their parents and descendants. In this paper, we aim to design mechanisms for general social networks structures and design mechanisms to incentivize invitations. Our first mechanism, Swap With Neighbors (SWN) restricts participants to choosing from their neighbors' items. Under this restriction, invitations can only bring matching opportunities not competitions, thus naturally satisfying IC. The drawback of SWN is that there is a very low probability that a participant can benefit from the enlarged market, which contradicts the purpose of enlarging the market for better matching. To improve the matching, we observed that if a group of participants trade and leave the market, they do not care about what happens next in the matching. Inspired by this observation, our second mechanism, called Leave and Share (LS), uses SWN as a base protocol but dynamically connects the remaining participants when some participants are leaving the market (the sharing process). Due to the sharing process, LS gives a more satisfying outcome.

All the above-mentioned mechanisms focused on the IC property, but to properly evaluate a matching mechanism, optimality and stability are the two key properties. The former characterizes how efficient a matching is, i.e., whether participants can improve their matching without making others' matching worse, and the latter depicts the robustness of a matching, i.e., whether participants have incentives to deviate from the matching and exchange within a smaller group. The Top Trading Cycles (TTC) presented by \citet{shapley1974cores} is the only Pareto optimal and stable mechanism for housing market problem~\cite{ma1994strategy}. However, \citet{kumar2022integration} show that more than half of the participants prefer their matching given by TTC in a smaller market than in a large integrated market. Hence, in the network setting where participants can control the market size through strategic invitations, TTC fails to be incentive compatible (IC). Due to the uniqueness of TTC, stability and Pareto optimality are not compatible with IC in housing market on networks. 

To maintain invitation incentives, we weaken stability and optimality notions with complete components (cc), called stable-cc and optimal-cc, respectively. Stable-cc weakens stability in the sense that only participants in a complete component can deviate from the matching and exchange within the component.
Optimal-cc weakens optimality in the sense that a matching can only be improved if all the participants whose allocation is improved have to form a complete component. We then prove that both stable-cc and optimal-cc are the tightest notions compatible with IC. Besides IC, individual rationality (IR) is also a commonly considered property that ensures participating in the game is not harmful, serving as a base incentive for people to join the game. Finally, combining all the properties together, we construct the theoretical boundaries for housing market on networks.

The best properties we can achieve together for our setting are stable-cc, optimal-cc, IC, and IR. The previously mentioned SWN and LS detect all the cycles formed by neighbors, so they naturally satisfy stable-cc. However, they both add restrictions to the selection space and fail to satisfy optimal-cc. Hence, we further propose the third mechanism, called Connected Trading Cycles (CTC), which does not add restrictions on the participant's selection space to satisfy all the properties. Similar to TTC, CTC detects the top trading cycles (to ensure optimality). The key difference is that CTC makes sure the trading cycles have to satisfy certain connectedness on networks (to ensure invitation incentives). Intuitively, the connectedness guarantees that a group of participants can stay together regardless of the others' strategic invitations. Thus, in the network setting, the trading cycles should build on a group of connected participants.

To sum up, our contributions are threefold:
\begin{itemize}
    \item We define the stability and optimality notions compatible with IC for housing market on networks. 
    \item We prove the best achievable stability and optimality in the network setting. 
    \item We propose three mechanisms: Swap With Neighbors, Leave and Share, and Connected Trading Cycles to meet the desirable properties (Table~\ref{tab:mechanism} gives a quick overview of the theoretical performance of each mechanism).
\end{itemize}

\begin{table}[ht]
\caption{Comparison on housing market mechanisms on social networks. We propose SWN, LS and CTC in this paper.}
\label{tab:mechanism}
\begin{tabular}{ccccc}\toprule
    \textit{Mechanism} & \textit{Stable-cc}& \textit{Optimal-cc} & \textit{IC} & \textit{IR}\\ \midrule
    TTC~\cite{shapley1974cores} &\ding{51} & \ding{51}& \ding{53}& \ding{51} \\
    Modified TTC~\cite{DBLP:conf/atal/KawasakiWTY21} & \ding{51}(Trees) &\ding{53} & \ding{51}(Trees)& \ding{51}  \\
    SWN & \ding{51} &\ding{53}& \ding{51}& \ding{51} \\
    LS  & \ding{51} &\ding{53}& \ding{51}& \ding{51} \\
    \textcolor{red}{CTC}  & \textcolor{red}{\ding{51}} &\textcolor{red}{\ding{51}} & \textcolor{red}{\ding{51}}& \textcolor{red}{\ding{51}} \\ \bottomrule
\end{tabular}
\end{table}

\section{Related Work}
\citet{li2017mechanism} initiate the line of research on mechanism design on social networks. They utilize social networks to incentivize participants to invite their friends to form a larger market so that participants can receive better outcomes~\cite{DBLP:conf/atal/Zhao21}. However, the participants are competitors and may not want to invite all their friends without proper incentives. Many classic solutions fail to provide such incentives, so we cannot directly apply them for this purpose. To combat this, in auctions, we can pay the harmed inviters some rewards~\cite{li2022diffusion}. In cooperative games, we can let the inviters share their invitees' contributions~\cite {zhang2022incentives}. However, in matching, we face a greater challenge since the mechanism cannot compensate participants for the loss caused by invitation through payments. 

Existing solutions for matching on networks restricts participants' selection space and the social network structures. \citet{DBLP:conf/atal/KawasakiWTY21} and \citet{tenants_social_network} model the social network as a directed graph where neighborhood relationships can be asymmetric (i.e., A is B's neighbor does not necessarily mean B is A's neighbor). They focus on trees and modify the classic mechanisms to incentivize invitations for housing market as well as its variant~\cite{abdulkadirouglu1999house}. Specifically, they add limits on the selection space of the participants in a tree to prevent an invitee competing with her inviters. For two-sided matching, \citet{ijcai2022-27} model the school choice problem into the network setting and design invitation incentives for the student side only. Another thread of work also takes into account the social network and its influence, but they differ from our setting~\cite{DBLP:conf/ijcai/GourvesLW17,HOEFER201320}. In their setting, the social network is priorly known and it constrains possible allocations.

Optimality and stability are two an well-concerned property for matching mechanisms~\cite{abdulkadiroglu2013matching}. The celebrated TTC is the only Pareto optimal and stable solution in the traditional housing market problem~\cite{ma1994strategy}. \citet{abraham2004pareto} study various ways to evaluate the optimality of a house allocation mechanism and illustrate why TTC gives the Pareto optimal allocation. \citet{fleischer2008dynamic} point out the transitions between different optimal allocations in the house allocation problem. \citet{brandt2019convergence} analyze whether different types of dynamic pairwise swaps converge to Pareto optimal allocations for matching markets. In the network setting, no previous work has investigated the theoretical boundaries. We propose tight stability and optimality notions for this new setting and design mechanisms to reach them.

\section{The Model}
We consider a housing market problem on a social network denoted by an undirected graph $G=(N,E)$, which contains $n$ agents $N=\{1,\dots,n\}$. Each agent $i\in N$ is endowed with an indivisible item $h_i$, usually referred to as a house. Let $H=\{h_1,\dots,h_n\}$ be the set of all agents' items. Considering the nature of the exchange economy, we formulate agents' social relationships as symmetric in this model, i.e., $i$ has connections with $j$ in the social network indicates $j$ has connections with $i$. We define agent $i$ as $j$'s neighbor if there is an edge $e\in E$ between agent $i$ and $j$, let $r_i\subseteq N$ be $i$'s neighbor set. 

Despite social relationship, each agent $i\in N$ has a strict preference $\succ_i$ over $H$. $h\succ_i h'$ means $i$ prefers $h$ to $h'$ and we use $\succeq_i$ to represent the weak preference. Thus, we denote agent $i$'s private type as $\theta_i=(\succ_i,r_i)$ and $\theta = (\theta_1,\cdots,\theta_n)$ as the type profile of all agents. Let $\theta_{-i}$ be the type profile of all agents except for agent $i$, then $\theta$ can be written as $(\theta_i,\theta_{-i})$. Let $\Theta$ be the type profile space of all agents. Similarly, we have $\Theta = (\Theta_i,\Theta_{-i})$.

In the housing market problem, the goal is to construct a matching following the principle that each agent will exchange their endowments to get better allocations. Assume there is a trusted center to enforce the execution of a matching mechanism, each agent is required to report her type (reporting neighbor set is treated as inviting neighbors in practice). We denote agent $i$'s reported type as $\theta'_i = (\succ'_i,r'_i)$, where $\succ'_i$ is the reported preference and $r'_i\subseteq r_i$ is the reported neighbor set. Let $\theta'=(\theta'_1,\cdots,\theta'_n)$ be the reported type profile of all agents. 

\begin{definition}
In a housing market problem, a matching mechanism is defined by an allocation policy $\pi = (\pi_i)_{i\in N}$, where $\pi_i:\Theta \to H$ satisfies for all $\theta \in \Theta$, $\pi_i(\theta) \in H$, and $\pi_i(\theta) \neq \pi_{j\neq i}(\theta)$.
\end{definition}

In the networked setting, we assume only a subset of the agents are initially in the game and the matching mechanism can incentivize agents to diffuse the information thus enlarging the matching. Without loss of generality, suppose an agent set $N_0 \subseteq N$ contains the initial participants in the market. The others need the existing participants' invitation to join the game. Since the invitation is modeled as reporting neighbors, we define the qualified participants by their reported types.

For a given reported type profile $\theta'$, we generate a directed graph $G(\theta')=(N(\theta'),E(\theta'))$, where edge $\langle i,j\rangle \in E(\theta')$ if and only if $j \in r_i'$. Under $\theta'$, we say agent $i$ is qualified if and only if there is a path from any agent in $N_0$ to $i$ in $G(\theta')$. That is, $i$ can connect to the agent set $N_0$ by an invitation chain. Let $Q(\theta')$ be the set of all qualified agents under $\theta'$. In the networked housing market problem, diffusion matching mechanisms can only use $Q(\theta')$.

\begin{definition}
In a networked housing market problem, a diffusion matching mechanism is a matching mechanism $\pi = (\pi_i)_{i\in N}$, such that for all reported type profile $\theta'$, it satisfies:
\begin{enumerate}
    \item for all unqualified agents $i\notin Q(\theta')$, $\pi_i(\theta') = h_i$.
    \item for all qualified agents $i\in Q(\theta')$, $\pi_i(\theta')$ is independent of the reports of all unqualified agents.
\end{enumerate}
\end{definition}

The difference between a diffusion matching and the matching defined in Definition 3.1 is that the participants can affect the qualification of other participants. If a participant changes her reported neighbor set, the qualified agent set may change. This is the challenge of this setting.

Next, we define two desirable properties for diffusion matching mechanisms: individual rationality and incentive compatibility. Intuitively, individual rationality requires that for each agent, reporting her type truthfully guarantees that she gets an item no worse than her own. For incentive compatibility, it means reporting type truthfully is a dominant strategy for each agent.

\begin{definition}[Individual Rationality (IR)]
A mechanism $\pi$ is individually rational if for all $i\in N $, all $\theta_i \in \Theta_i$, and all $\theta'_{-i}\in \Theta_{-i}$, we have $\pi_i(\theta_i,\theta_{-i}')\succeq_i h_i$.
\end{definition}

\begin{definition}[Incentive Compatibility (IC)]
A mechanism $\pi$ is incentive compatible if for all $i\in N$, all $\theta'_{-i}\in \Theta_{-i}$ and all $\theta_i, \theta'_i\in \Theta_i$, we have $\pi_i(\theta_i,\theta'_{-i})\succeq_i \pi_i(\theta'_i,\theta'_{-i})$.
\end{definition}

To evaluate the performance of a matching mechanism, an important metric is called Pareto optimality. A matching is Pareto optimal if no agent can improve her matching without others' allocation getting worse.

\begin{definition}[Pareto Optimality (PO)]
A mechanism $\pi$ is Pareto optimal if for all type profile $\theta$, there is no other allocation $\pi'(\theta)$ such that for each qualified agent $i\in Q(\theta)$, $\pi'_i(\theta) \succeq_i \pi_i(\theta)$, and there exists at least one qualified agent $j\in Q(\theta)$, $\pi'_j(\theta) \succ_j \pi_j(\theta)$.
\end{definition}

Another metric is stability. It requires no group of agents can deviate from the matching and match among themselves to make no one in the group worse off and at least one better off. 

\begin{definition}[Stability]
A mechanism $\pi$ is stable if for all type profile $\theta$, there is no other allocation $\pi'(\theta)$ such that there exists a group $S\subseteq Q(\theta)$ that satisfies $\forall i\in S, \pi'_i(\theta)\in H_S$, and $\forall i\in S, \pi'_i\succeq_i\pi_i(\theta)$ with at least one $j\in S, \pi'_j(\theta)\succ_j\pi_j(\theta)$.
\end{definition}

In summary, our model is a practical variant of the networked housing market model proposed in~\cite{DBLP:conf/atal/KawasakiWTY21}. In their paper, the authors designed matching on an asymmetric social network and proved the impossibility results of the compatibility of IC, PO and stability. In the next section, we will prove the impossibilities in our model and further construct the theoretical boundaries.

\section{The Boundaries}
In this section, we first prove that a diffusion matching mechanism cannot achieve Pareto optimality or stability together with IC. Next, we define the best optimality and stability notion compatible with IC and then construct the theoretical boundaries.

\begin{table}[t]
	\caption{Desirable properties for diffusion matching mechanisms. IC, IR, stable-cc and optimal-cc are the boundaries.The implications between properties are proved in Section 4.3.}
	\label{tab:boundaries}
	\begin{tabular}{cccc}\toprule
		\textit{IR} & \textit{IC} & \textit{Stability} & \textit{Optimality} \\ \midrule
		\textcolor{red}{IR} & \textcolor{red}{IC} & Stable & PO \\
		 &  & $\downarrow$ & $\downarrow$ \\
		 &  & Stable-wcc & Optimal-wcc \\
	   &  & $\downarrow$ & $\downarrow$ \\
          &  & \textcolor{red}{Stable-cc} & \textcolor{red}{Optimal-cc} \\ \bottomrule
	\end{tabular}
\end{table}

\subsection{Impossibility Results}

\begin{theorem}[Impossibility for PO, IC and IR]
No diffusion matching mechanism is PO, IC and IR.
\label{the:POxIC}
\end{theorem}
\begin{proof}
In the example shown in Figure \ref{fig:POstaxIC}, suppose $N_0=\{1,2\}$, the Pareto optimal and IR allocations are $\pi_1=(h_3,h_2,h_1)$ and $\pi_2=(h_2,h_1,h_3)$. If a mechanism allocates $\pi_1$, agent 2 can misreport her neighbor set as $r'_2=\{1\}$. Under agent 2’s misreport, agent 3 cannot join the game and the only PO and IR allocation will be $\pi_2=(h_2,h_1,h_3)$, and 2 gets a better allocation compared to that in $\pi_1$. However, if the mechanism allocates $\pi_2$, agent 1 can misreport her preference as $h_3\succ'_1 h_1 \succ'_1 h_2$. In this way, the only PO and IR allocation is $\pi_1=(h_3,h_2,h_1)$, and agent 1 reaches a better allocation. Hence, no diffusion matching mechanism is PO, IC and IR. 
\end{proof}

\begin{figure}[ht]
\centering
\begin{minipage}[t]{0.45\linewidth}
\centering
\begin{tikzpicture}[scale=0.2, sibling distance=5em,
  every node/.style = {scale=0.65, shape=circle, draw, align=center},
    outline/.style={draw=#1,thick,fill=#1!100}]
  every draw/.style = {scale=1}
  \node[] (node1) at (0,0) {1};
  \node[] (node2) at (4,0) {2};
  \node[] (node3) at (0,4) {3};
  \draw[latex-latex] (node1)--(node2);
  \draw[latex-latex] (node2)--(node3);
\draw[latex-] (node1)--(node3);
\end{tikzpicture}
\end{minipage}
\caption{A social network example. Preferences are $h_3 \succ_1 h_2 \succ_1 h_1, \ \ \ h_1 \succ_2 h_2 \succ_2 h_3, \ \ \ h_1 \succ_3 h_3 \succ_3 h_2$.} 
\label{fig:POstaxIC}
\Description{Example for Theorem 4.1 and Theorem 4.2.}
\end{figure}
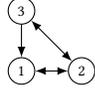

\begin{theorem}[Impossibility for stability and IC]
No diffusion matching mechanism is stable and IC.
\label{the:StaxIC}
\end{theorem}
\begin{proof}
Consider the example in Figure~\ref{fig:POstaxIC}, suppose $N_0=\{1\}$, the only stable allocations is $\pi_1=(h_3,h_2,h_1)$. However, agent 2 can misreport her neighbor set as $r'_2=\{1\}$, so agent 3 cannot join the game. In this way, the only stable allocation is $\pi_2=(h_2,h_1,h_3)$. This means agent 2 can misreport to improve her matching result, so stability is incompatible with IC.
\end{proof}

To seek achievable optimality and stability notions for IC diffusion matching mechanisms, we focus on a special graph structure: complete components, i.e., the agents are neighbors of each other. For Pareto optimality, we restrict that only fully connected agents can improve their matching. For stability, we constrain the agents who can deviate from the matching and swap among themselves to be fully connected.

\begin{definition}[Complete Component]
A connected directed graph $G=(V,E)$ is a complete component if for any two nodes $i,j\in V$, we have $\langle i,j\rangle \in E$. 
\end{definition}

The idea of the limitation comes from the reality that only a group of people who know each other have a higher possibility to communicate and negotiate for a better matching/trade. Besides, in the traditional housing market problem, since there are no constraints on social connections, agents can be viewed as fully connected. Then, any group of agents is a fully connected component.

Following this idea, we define optimality and stability under complete components: optimal-cc and stable-cc.

\begin{definition}[Optimality under Complete Components (Optimal-cc)]
A diffusion matching mechanism $\pi$ is optimal under complete components if for all type profiles $\theta$ and the allocation $\pi(\theta)$, there is no other allocation $\pi'(\theta)$ such that $\forall i \in N, \pi'_i(\theta) \succeq_i \pi_i(\theta)$ and $\exists j \in N, \pi'_j(\theta) \succ_j \pi_j(\theta)$ and agents $\{i\in N|\pi_i(\theta) \neq \pi'_i(\theta)\}$ forms a \textbf{complete component} in $G(\theta)$. 
\end{definition}

\begin{definition}[Stability under Complete Components (Stable-cc)]
A diffusion matching mechanism $\pi$ is stable under complete components if for all type profiles $\theta$ and the allocation $\pi(\theta)$, there is no agent set $S\subseteq N$ (with item set $H_S\subseteq H$) that forms a \textbf{complete component} in $G(\theta)$, and another allocation $\pi'(\theta)$ with $\forall i \in S, \pi'_i(\theta) \in H_S$ such that $\forall i \in S, \pi'_i(\theta)\succeq_i\pi_i(\theta)$ and $\exists j \in S, \pi'_j(\theta)\succ_j\pi_j(\theta)$.
\end{definition}

\subsection{Stable-cc and Optimal-cc are the Boundaries}
Stable-cc and optimal-cc pose strict constraints on the agents who can improve by exchanging internally. Hence, it is worth investigating whether we can relax this limitation to get stronger notions. For instance, can we allow any connected group (no need to be a complete component) to improve their matching? Following this idea, we define stronger notions based on weakly complete components, but they are not compatible with IC.

\begin{definition}[Weakly Complete Component]
A connected directed graph $G=(V,E)$ is a weakly complete component if there exists at most one pair of nodes $i,j\in V$ such that $\langle i,j\rangle\notin E$.
\end{definition}

\begin{definition}[Optimality under Weakly Complete Components (Optimal-wcc)]
A diffusion matching mechanism $\pi$ is optimal under weakly complete components if for all type profiles $\theta$ and the allocation $\pi(\theta)$, there is no other $\pi'(\theta)$ such that $\forall i \in N, \pi'_i(\theta) \succeq_i \pi_i(\theta)$ and $\exists j \in N, \pi'_j(\theta) \succ_j \pi_j(\theta)$ and agents $\{i\in N|\pi_i(\theta) \neq \pi'_i(\theta)\}$ forms a \textbf{weakly complete component} in $G(\theta)$. 
\end{definition}

Next, we prove that no IC and IR diffusion mechanisms can be optimal-wcc.

\begin{figure}[t]
\centering
\subfigure{
\centering
\begin{minipage}[t]{1\linewidth}
\centering
\begin{tikzpicture}[scale=0.2, sibling distance=5em,
  every node/.style = {scale=0.65, shape=circle, draw, align=center},
    outline/.style={draw=#1,thick,fill=#1!100}]
  every draw/.style = {scale=1}

  \node[] (node1) at (0,0) {1};
  \node[] (node2) at (4,0) {2};
  \node[] (node3) at (0,4) {3};
  \node[] (node4) at (4,-4) {4};
  \draw[latex-latex] (node1)--(node3);
  \draw[latex-latex] (node1)--(node2);
  \draw[latex-] (node1)--(node4);
  \draw[latex-] (node2)--(node3);
  \draw[latex-latex] (node2)--(node4);

  \draw[-latex,red] (node1).. controls(1,-3) ..(node4);
  \draw[-latex,red] (node2).. controls(3,3) ..(node3);
  \draw[-latex,red] (node3).. controls(-1,2) ..(node1);
  \draw[-latex,red] (node4).. controls(3,-1) ..(node1);
  \draw[-latex,red,dashed] (node1).. controls(2,1) ..(node2);
  
  \node[rectangle,color=black,draw=none] (node002) at (2,-6) {(a)};



  \node[] (node1a) at (14,0) {1};
  \node[] (node2a) at (18,0) {2};
  \node[] (node3a) at (14,4) {3};
  \node[dashed] (node4a) at (18,-4) {4};
  \draw[latex-latex] (node1a)--(node2a);
  \draw[latex-] (node2a)--(node3a);
  \draw[latex-latex] (node3a)--(node1a);

  \draw[-latex,red] (node2a).. controls(17,3) ..(node3a);
  \draw[-latex,red] (node3a).. controls(13,2) ..(node1a);
  \draw[-latex,red,dashed] (node1a).. controls(16,1) ..(node2a);
  
  \node[rectangle,color=black,draw=none] (node002) at (16,-6) {(b)};


  \node[] (node1b) at (28,0) {1};
  \node[] (node2b) at (32,0) {2};
  \node[dashed] (node3b) at (28,4) {3};
  \node[] (node4b) at (32,-4) {4};
  \draw[latex-latex] (node1b)--(node2b);
  \draw[latex-] (node1b)--(node4b);
  \draw[latex-latex] (node2b)--(node4b);

  \draw[-latex,red] (node1b).. controls(29,-2) ..(node4b);
  \draw[-latex,red] (node4b).. controls(31,-1) ..(node1b);
  \draw[-latex,red,dashed] (node1b).. controls(30,1) ..(node2b);
  
  \node[rectangle,color=black,draw=none] (node002) at (30,-6) {(c)};

\end{tikzpicture}
\end{minipage}
}
\subfigure{
\centering
\begin{minipage}[ht]{1\linewidth}
\centering
\small
\begin{tabular}{cc}\toprule
    \textit{Example} & \textit{Optimal-wcc Allocation}\\ \midrule
    (a) & $\pi_a=(h_4,h_2,h_3,h_1)$, $\pi_b=(h_2,h_3,h_1,h_4)$\\
    (b) & $\pi_c=(h_2,h_3,h_1,h_4)$\\
    (c) & $\pi_d=(h_4,h_2,h_3,h_1)$\\ \bottomrule
\end{tabular}
\end{minipage}
}
\caption{A counterexample for the coexistence of optimal-wcc, IR, and IC. Preferences are $h_4 \succ_1 h_2 \succ_1 h_1, \ \ \ h_3 \succ_2 h_2, \ \ \ h_1 \succ_3 h_3, \ \ \ h_1 \succ_4 h_4$. Agents 1,2 are initial players in the matching. The red solid arrows represent favorite pointing. The red dashed arrows represent the second favorite pointing. The dashed agents are unqualified, so they are allocated their endowments.}
\label{fig:counterexample}
\Description{A counterexample for the coexistence of optimal-wcc, IR, and IC.}
\end{figure}
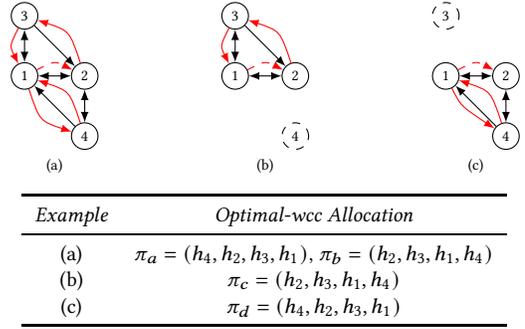

\begin{theorem}
No IC and IR diffusion one-sided matching mechanism is optimal-wcc.
\label{the:optwcc}
\end{theorem}
\begin{proof}
Let's consider the example given in Figure~\ref{fig:counterexample}. In figure (a), IR and optimal-wcc solutions are $\pi_a,\pi_b$. If we choose $\pi_a$, agent 2 cannot get $h_3$. So, agent 2 has the incentive to misreport and not invite agent 4. In this way, figure (a) transforms to figure (b) where the only IR and optimal-wcc solution is $\pi_c$. Then, agent 2 can get her favorite item $h_3$. To ensure IC, we should allocate $h_3$ to agent 2 when she truthfully reports (i.e., allocate $\pi_b$ in figure (a)). Next, we consider agent 1's incentive to misreport when we choose $\pi_b$ in figure (a). If agent 1 does not invite agent 3, figure (a) transforms to figure (c) where the only IR and optimal-wcc solution is $\pi_d$. Here, agent 1 can get her favorite item $h_4$. To ensure IC in figure (a), we should allocate $h_4$ to agent 1 when she truthfully reports (i.e., allocate $\pi_a$ in figure (a)). This is a contradiction, so optimal-wcc is not compatible with IC. Hence, the tightest optimality notion compatible with IC is optimal-cc.
\end{proof}

Similarly, for stability, we define a stronger notion called stable-wcc and prove that stable-wcc is not compatible with IC and IR. 

\begin{definition}[Stability under Complete Components (Stable-wcc)]
A diffusion matching mechanism $\pi$ is stable under weakly complete components if for all type profiles $\theta$ and the allocation $\pi(\theta)$, there is no agent set $S\subseteq N$ (with item set $H_S\subseteq H$) that forms a \textbf{weakly complete component} in $G(\theta)$, and another allocation $\pi'(\theta)$ with $\forall i \in S, \pi'_i(\theta) \in H_S$ such that $\forall i \in S, \pi'_i(\theta)\succeq_i\pi_i(\theta)$ and $\exists j \in S, \pi'_j(\theta)\succ_j\pi_j(\theta)$.
\end{definition}

\begin{theorem}
No IC and IR diffusion one-sided matching mechanism is stable-wcc.
\label{the:stablewcc}
\end{theorem}
\begin{proof}
Consider the example in Figure~\ref{fig:POstaxIC}, suppose $N_0=\{1\}$, the only stable-wcc allocations is $\pi_1=(h_3,h_2,h_1)$. However, agent 2 can misreport her neighbor set as $r'_2=\{1\}$, so that agent 3 cannot join the game. In this way, the only stable-wcc allocation is $\pi_2=(h_2,h_1,h_3)$. This means agent 2 can misreport to improve her matching result, so it can be concluded that stable-wcc is incompatible with IC.
\end{proof}

Combining Theorem~\ref{the:optwcc} and Theorem~\ref{the:stablewcc}, we conclude that the best stability and optimality we can aim for an IC and IR diffusion matching mechanism are stable-cc and optimal-cc respectively.

\subsection{Implications between Properties}
In this section, we study the relationships between different stability and optimality notions. By definition, Pareto optimality implies optimal-wcc and further implies optimal-cc. Stability implies stable-wcc and further implies stable-cc. In the traditional setting, a stable matching mechanism naturally satisfies Pareto optimal.

\begin{theorem}
A mechanism $\pi$ is stable implies that $\pi$ is PO.
\end{theorem}
\begin{proof}
For all type profiles $\theta$, if a stable mechanism $\pi$ is not Pareto optimal, there exists a $\pi'$ such that $\forall i\in N$, $\pi'_i(\theta) \succeq_i \pi_i(\theta)$ and $\exists j\in N$, $\pi'_j(\theta) \succ_j \pi_j(\theta)$. In this case, we can construct a group of agents $S$ starting from $\{j\in N| \pi'_j(\theta) \succ_j \pi_j(\theta)\}$ and consecutively add other agents so that $\forall i \in S$, $\pi'_i(\theta)\in H_S$. So, $S$ can deviate from the matching together which violates stability. Thus, $\pi$ is stable implies that $\pi$ is Pareto optimal.
\end{proof}

Thus, when designing mechanisms, we can only focus on stability, and the optimality condition will be satisfied at the same time. However, in the networked setting, a similar implication relationship does not exist between the redefined notions.

\begin{theorem}
A mechanism $\pi$ is stable-cc does not imply that $\pi$ is optimal-cc.
\end{theorem}
\begin{proof}
See the example in Figure~\ref{fig:figure1}, allocation $\pi=(h_2,h_1,h_4,h_3)$ is stable under complete components. However, there exists another allocation $\pi'=(h_2,h_4,h_1,h_3)$ where $\{2,3\}$ is a strictly better off group. Thus, $\pi$ is stable-cc does not imply that $\pi$ is optimal-cc.\end{proof}

\begin{figure}[ht]
\centering
\begin{tikzpicture}[scale=0.2, sibling distance=5em,
  every node/.style = {scale=0.65, shape=circle, draw, align=center},
    outline/.style={draw=#1,thick,fill=#1!100}]
  every draw/.style = {scale=1}
  \node[] (node1) at (0,0) {1};
  \node[] (node2) at (4,0) {2};
  \node[] (node3) at (8,0) {3};
  \node[] (node4) at (12,0) {4};
  \draw[latex-latex] (node1)--(node2);
  \draw[latex-latex] (node2)--(node3);
  \draw[latex-latex] (node3)--(node4);
\end{tikzpicture}
\caption{Preferences are $h_3 \succ_1 h_2 \succ_1 h_1, \ \ \ h_4 \succ_2 h_1 \succ_2 h_2, \ \ \ h_1 \succ_3 h_4 \succ_3 h_3, \ \ \ h_2 \succ_4 h_3 \succ_4 h_4$. }
\label{fig:figure1}
\Description{Example for stable-cc cannot imply optimal-cc.}
\end{figure}
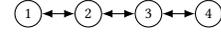

This vanish of implication also poses a greater challenge in designing a mechanism that reaches the new boundaries.

\section{The Mechanisms}
In this section, we introduce three IC mechanisms, Swap With Neighbors (SWN), Leave and Share (LS) and Connected Trading Cycles (CTC). By different techniques, they all eliminate competitions caused by invitation. That is to say, they guarantee that inviting more agents to join in only potentially improves one's matching. Formal proofs and complexity analysis are in Appendix B.

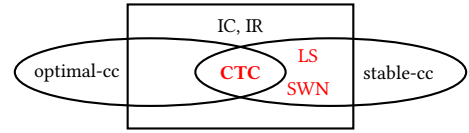
\begin{figure}[ht]
\centering
\begin{tikzpicture}[scale=0.3, sibling distance=5em,
  every node/.style = {scale=0.83, shape=circle, draw=none, align=center},
    outline/.style={draw=#1,thick,fill=#1!100}]
  every draw/.style = {scale=1}
  \draw[draw=black, thick] (5,6) rectangle (15,11.5);
  \draw[color=black,thick] (14,8.5) ellipse (6 and 1.5);
  \draw[color=black,thick] (6,8.5) ellipse (6 and 1.5);

  \node[rectangle,color=black] (node7) at (10,8.5) {\textcolor{red}{\textbf{CTC}}};
  \node[rectangle,color=black] (node9) at (13,9.25) {\textcolor{red}{LS}};
  \node[rectangle,color=black] (node10) at (13,7.75) {\textcolor{red}{SWN}};
  \node[rectangle,color=black] (node11) at (10,10.5) {IC, IR};
  \node[rectangle,color=black] (node12) at (17,8.5) {stable-cc};
  \node[rectangle,color=black] (node13) at (2.75,8.5) {optimal-cc};
  
\end{tikzpicture}
\caption{The properties of our three mechanisms. LS and SWN are stable-cc, IC and IR. CTC reaches the boundaries.}
\Description{A figure that shows the properties of our three mechanisms.}
\label{fig:Venn}
\end{figure}

\subsection{Swap With Neighbors}
First, we present SWN which is an intuitive extension of Top Trading Cycles. SWN only allows agents to choose favorite items from neighbors to achieve invitation incentives.

\begin{definition}[Swap With Neighbors]
For a given $G(\theta')$, construct a directed graph by letting each agent point to her favorite item among herself and her neighbors remaining in the matching. There is at least one cycle. For each cycle, allocate the item to the agent who points to it and remove the cycle. Repeat the process until there is no agent left.  
\end{definition}

In SWN, agents' allocation is determined by trading cycles, which makes it strategy-proof for the preference report. Also, each agent can only get allocated a house from her neighbors, which means misreporting on one's neighbor set is not beneficial, so SWN satisfies IC. Besides, a trading cycle within neighbors can always be allocated by SWN, so SWN satisfies stable-cc. 

However, SWN achieves IC and stable-cc at the cost of efficiency, as it limits the matching options for agents. To improve the matching, we should provide more matching opportunities. But this also brings agents' strategic report aiming for a better match. Thus, we need a strategy-proof order to help with the mechanism design and ensure that truthfully reporting one's type is a dominant strategy for every agent. Our order is based on the shortest path length from the initial agent set to each agent. The reason is that non-initial agents require other agents' permission to join the matching, the layer structure of the social network should be kept properly for incentive compatibility. This order is applied in both LS and CTC.



\begin{definition}
An ordering of agents is a one-to-one function $\mathcal {O}:\mathbb {N}^{+}\to N$, where agent $\mathcal {O}(t)$ is the $t^{th}$ agent in the ordering. For simplicity, we denote agent $i$'s order as $\mathcal{O}^{-1}(i)$. Agents in $\mathcal{O}$ are sorted in ascending order by the length of the shortest path from agent set $N_0$ to them. Especially, for any agent $i \in N_0$, its shortest path length from agent set $N_0$ is $0$. When multiple agents have the same length of the shortest path, we use a random tie-breaking. 
\label{def:order}
\end{definition}


\subsection{Leave and Share}


Leave and Share uses SWN as a base and adds a natural sharing process to enlarge agents' selection space, trying to provide a better allocation. The mechanism consists of two steps: first, agents are matched and left by rounds in a protocol that resembles SWN under a strategy-proof order. This guarantees that inviters are not worse off. Second, we share the neighbors of the left agents in this round by connecting their neighbors to each other, thus their neighbors can have new neighbors in the next round. This dynamic neighbor set update comes naturally because a matched cycle does not care how the remaining neighbors will be matched. Also, their remaining neighbors and other agents cannot prevent this sharing.

\begin{figure}[ht]
\centering

\begin{tikzpicture}[scale=0.2, sibling distance=5em,
  every node/.style = {scale=0.65, shape=circle, draw, align=center},
    outline/.style={draw=#1,thick,fill=#1!100}]
  every draw/.style = {scale=1}
  \node[] (node2) at (0,0) {1};
  \node[] (node3) at (4,0) {2};
  \node[] (node4) at (8,0) {3};
  \node[] (node5) at (12,0) {4};
  \draw[latex-latex] (node2)--(node3);
  \draw[latex-latex] (node3)--(node4);
  \draw[latex-latex] (node4)--(node5);
  \draw[-latex,red] (node3).. controls(6,1) ..(node4);
  \draw[-latex,red] (node4).. controls(6,-1) ..(node3);
  \draw[-latex,red] (node2).. controls(4.5,2.5)..(node5);
  \draw[-latex,red] (node5).. controls(7.5,-2.5)..(node2);

\end{tikzpicture}

\caption{Preferences are $h_4 \succ_1 h_1 \succ_1 \cdots, \ \ \ h_3 \succ_2 h_2 \succ_2 \cdots, \ \ \ h_2 \succ_3 h_3 \succ_3 \cdots, \ \ \ h_1 \succ_4 h_4 \succ_4 \cdots$. 
}
\label{fig:value_ls}
\Description{Driving example that shows the value of \textit{Leave and Share}.}
\end{figure}
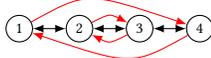

To see the value of our mechanism, consider the example given in Figure~\ref{fig:value_ls}, where only agents 2 and 3 can exchange with each other in SWN. The rest of the agents will end up with their own items. However, agents 2 and 3 will not block the exchange for agents 1 and 4 once they get their preferred items. After agents 2 and 3 are matched and \textbf{Leave}, we \textbf{Share} their remaining neighbors then agents 1 and 4 can swap. The process of Leave and Share is the name and core of our mechanism.

Before formalizing our mechanism, we introduce two notations to simplify our description.

\begin{definition}
Given a set $A \subseteq N$, we say $f_i(A)= j\in A$ is $i$'s favorite agent in $A$ if for any agent $k\in A, h_j\succeq'_i h_{k}$. 
\end{definition}

\begin{framed}
 \noindent\textbf{Leave and Share (LS)}
  
 \begin{enumerate}
  \item Initialize $N_{out}=\emptyset$ and an empty stack $S$. Define the top and bottom of $S$ as $S_{top}$ and $S_{bottom}$ respectively, and let $R_{i}=r_{i}'\cup \{ S_{bottom} , i\}$.
  \item While $N_{out} \neq N$:
  \begin{enumerate}
    \item Find the minimum $t$ such that $\mathcal{O}(t) \notin N_{out}$. Push $\mathcal{O}(t)$ into $S$.

    \item While $S$ is not empty:
        \begin{enumerate}
            \item While $f_{S_{top}}(R_{S_{top}}) \notin S$, push $f_{S_{top}}(R_{S_{top}})$ into $S$.
            \item Pop off all agents from $S_{top}$ to $f_{S_{top}}(R_{S_{top}})$, who already formed a trading cycle $C$ following their favorite agents. Allocate each agent $i\in C$ the item $h_{f_i(R_i)}$ . Add $C$ to $N_{out}^t$.
            \item Update the neighbor set of $C$'s remaining neighbors by removing $C$, i.e., for all $j\in \bigcup_{i\in C} r_i' \setminus N_{out}^t$, set $r_j' = r_j' \setminus C$. 
         \end{enumerate}
    \item Add $N_{out}^t$ to $N_{out}$. Let all remaining neighbors of $N_{out}^t$ connect with each other, i.e., they become neighbors of each other. That is, let $X = \bigcup_{i\in N_{out}^t} r_i' \setminus N_{out}^t$ and for all $j\in X$, set $r_j' = r_j' \cup X$. 
    \end{enumerate}
  
 \end{enumerate}
\end{framed}

In LS, we first define an order $\mathcal {O}$ which depends on each agent's shortest distance to the initial agent set. Under this order, the first while loop (step $2$) guarantees that the agent pushed into the stack is the remaining agent with the smallest order, and all agents are matched (including self-match) in the end. A new round begins each time the stack empties. 

In the Leave stage, each agent that is pushed into the stack pushes her (current) favorite neighbor into the stack (step (a)). If her favorite is already in the stack, we pop all the agents between herself and her favorite to form a trading cycle. Specially, we allow agents to choose the agent at the bottom of the stack as favorite, which leads to popping off all the agents in the stack (step (b)).

Once the stack is empty, the mechanism enters the Share stage and updates the neighbor set of the remaining agents (step (c)). All neighbors of the left agents become new neighbors to each other. In the next Leave stage, they can choose 
in a larger neighbor set.

\subsection{Connected Trading Cycles}
In this section, we present our mechanism called Connected Trading Cycles (CTC) which takes both the trading cycles and agents' connections into consideration. CTC is IR, IC, optimal-cc, and stable-cc, which is the best we can get in the network setting. To begin with, we give a few definitions to serve the description of our mechanism.

\begin{definition}
Given a reported type profile $\theta'$, we generate a directed graph $F(\theta')$, in which each qualified agent has exactly one outgoing edge pointing to a qualified agent in $G(\theta')$. An edge $\langle i,j\rangle$ in $F(\theta')$ means agent $i$ likes agent $j$'s item the most among all the qualified agents' items ($i,j$ can be the same agent). We name $F(\theta')$ as the favorite pointing graph for $\theta'$. 
\end{definition}

There is at least one cycle in $F(\theta')$, and we formalize the definition of a cycle as below.

\begin{definition}
A cycle in $F(\theta')$ is an agent sequence $C_m = \{c_1,\dots,c_m\}$ such that $\forall i\in \{1,\dots,m-1\}$, $c_i$ points to $c_{i+1}$, and $c_m$ points to $c_1$ in the favorite pointing graph $F(\theta')$.  
\end{definition}


Since the optimality and stability boundaries are constructed on complete components, it is natural to let cycles build on complete components to trade. However, if we only allow cycles formed by complete components to get traded, the performance is not better than SWN or LS (it meets stable-cc, but far from optimal-cc). Thus, we also consider cycles formed by connected components. Given that optimal-c is not compatible with IC, it is clear that not all such cycles can get traded. Therefore, we should further distinguish these cycles and identify the ones compatible with IC. This detection also makes our CTC description complex.

\begin{figure}[ht]
\centering
\begin{tikzpicture}[scale=0.2, sibling distance=5em,
  every node/.style = {scale=0.65, shape=circle, draw, align=center},
    outline/.style={draw=#1,thick,fill=#1!100}]
  every draw/.style = {scale=1}
  \node[] (node4) at (0,0) {4};
  \node[] (node2) at (4,0) {2};
  \node[] (node1) at (8,0) {1};
  \node[] (node3) at (12,0) {3};
  \node[] (node5) at (16,0) {5};
  \draw[latex-latex] (node4)--(node2);
  \draw[latex-latex] (node2)--(node1);
  \draw[latex-latex] (node1)--(node3);
  \draw[latex-latex] (node3)--(node5);
  \draw[-latex,red,dashed] (node1).. controls(6,2) ..(node2);
  \draw[-latex,red] (node3).. controls(6,3) ..(node4);
  \draw[-latex,blue] (node1).. controls(12,2) ..(node5);
  \draw[-latex,blue] (node5).. controls(14,-2) ..(node3);
  \draw[-latex,blue,dashed] (node3).. controls(10,-2) ..(node1);
  \draw[-latex,red] (node4).. controls(4,-2) ..(node1);
  \draw[-latex,red] (node2).. controls(8,-2) ..(node3);

  \draw[-latex] (4,-2)--(0,-5);
  \node[draw=none,font=\fontsize{12}{6}\selectfont] (node20) at (-4.5,-3) {3 does not invite 5};

  \node[] (node04) at (-12,-7) {4};
  \node[] (node02) at (-8,-7) {2};
  \node[] (node01) at (-4,-7) {1};
  \node[] (node03) at (0,-7) {3};
  \node[dashed] (node05) at (4,-7) {5};
  \draw[latex-latex] (node04)--(node02);
  \draw[latex-latex] (node02)--(node01);
  \draw[latex-latex] (node01)--(node03);
  \draw[-latex,red,dashed] (node01).. controls(-6,-5) ..(node02);
  \draw[-latex,red] (node03).. controls(-6,-4) ..(node04);
  \draw[-latex,red] (node04).. controls(-8,-9) ..(node01);
  \draw[-latex,red] (node02).. controls(-4,-9) ..(node03);

  \draw[-latex] (12,-2)--(16,-5);
  \node[draw=none,font=\fontsize{12}{6}\selectfont] (node21) at (20.5,-3) {1 does not invite 2};

  \node[dashed] (node14) at (11,-7) {4};
  \node[dashed] (node12) at (15,-7) {2};
  \node[] (node11) at (19,-7) {1};
  \node[] (node13) at (23,-7) {3};
  \node[] (node15) at (27,-7) {5};
  \draw[latex-latex] (node11)--(node13);
  \draw[latex-latex] (node13)--(node15);
  \draw[-latex,blue] (node11).. controls(23,-5) ..(node15);
  \draw[-latex,blue] (node15).. controls(25,-9) ..(node13);
  \draw[-latex,blue,dashed] (node13).. controls(21,-9) ..(node11);
  
\end{tikzpicture}
\caption{Preferences are $h_5 \succ_1 h_2 \succ_1 h_1, \ \ \ h_3 \succ_2 h_2, \ \ \ h_4 \succ_3 h_1 \succ_3 h_3, \ \ \ h_1 \succ_4 h_4, \ \ \ h_3 \succ_5 h_5$. The solid arrows represent favorite pointing. The dashed arrows represent second favorite pointing. The dashed agents are unqualified agents
. }
\label{fig:figure_cycle}
\Description{A motivated example for \textbf{connected cycles}.}
\end{figure}
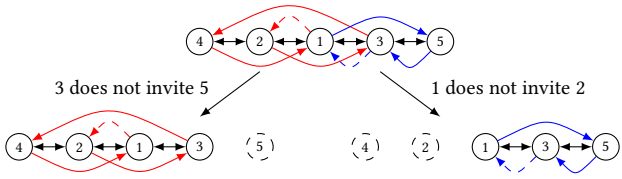

Here comes a motivated example. The social network is shown in Figure~\ref{fig:figure_cycle} and agent 1 is the initial agent. If agent 1 does not invite agent 2, the connected component $\{1,3,5\}$ will form the blue cycle ($C_3=\{1,5,3\}$). On the other hand, if agent 3 does not invite agent 5, the connected component $\{4,2,1,3\}$ will form the red cycle ($C_4=\{1,2,3,4\}$). If both cycles are allowed, IC cannot hold since both agent 1 and agent 3 have the incentive to not invite a neighbor and form the cycle in favor of herself. Thus, an IC mechanism can only allow at most one of these two cycles to get traded. To distinguish the two cycles, we pay attention to each agent's pointing. In the blue cycle, each agent has an exclusive path to reach her pointing: $\textbf{1}\to 3,\textbf{3}\to 5, \textbf{5}\to 3\to 1$. However, in the red cycle, $\textbf{2}\to 1\to 3$ and $\textbf{4}\to 2\to 1$ have a shared edge $2\to 1$. Similarly, $\textbf{1}\to 2$ and $\textbf{3}\to 1\to 2\to 4$ have a shared edge $1\to 2$. Recall that in the cycle formed by a complete component (which is always allowed to get traded), each agent has an exclusive path to her pointing. Intuitively, the exclusive path can be regarded as an insurance for invitation incentives (no agent wants to influence other's pointing by not inviting a neighbor). Following this idea, we design an algorithm to detect cycles in which each agent has an exclusive path to her pointing.

We also needs a strategy-proof order of the agents to determine who can be matched first.


\begin{framed}
\noindent\textbf{Path Detection}

\noindent Input: a cycle $C_m$, a reported type profile $\theta'$, a favorite pointing graph $F(\theta')$ and social network $G(\theta')$.
\begin{enumerate}
    \item [a.] Find a minimum connected component $T$ in $G(\theta')$ where $C_m\subseteq T$ and $\forall i\in T$, let $j$ be $i$'s pointing in $F(\theta')$, we have $j\in T$. Note that if the cycle $C_m$ form a connected component in $G(\theta')$, then $T=C_m$. If there is no such $T$ let $T = \emptyset$.
    \item [b.] Construct a subgraph $G_T$ with the reported type profile of all agents in $T$.
    \item [c.] If agent $i\in T$ points to herself in $F(\theta')$, add a mark $i$ on all the outgoing edges of $i$ in $G_T$.
    \item [d.] Construct a shortest path set $SP=\{sp_i|i\in T\}$ in $G_T$, where $sp_i$ is the shortest path from agent $i\in T$ to her pointing in $F(\theta')$. While $SP\neq \emptyset$, execute the following,
    \begin{enumerate}
        \item [i.] Sort $SP$ by the ascending order of path length. If $sp_i$ and $sp_j$ have the same length, sort them by the agents ascending order $\mathcal{O}^{-1}(i)$ and $\mathcal{O}^{-1}(j)$.
        \item [ii.] Remove the first $sp_i$ in $SP$. Add a mark $i$ on every edge on path $sp_i$ in $G_T$. If $sp_i$ contains marks other than $i$, and $i$ has another path to her pointing, add the next shortest path $sp'_i$ for $i$ to $SP$.
    \end{enumerate}
\end{enumerate}
\noindent Output: an agent set $T$ and the marked subgraph $G_T$.
\end{framed}

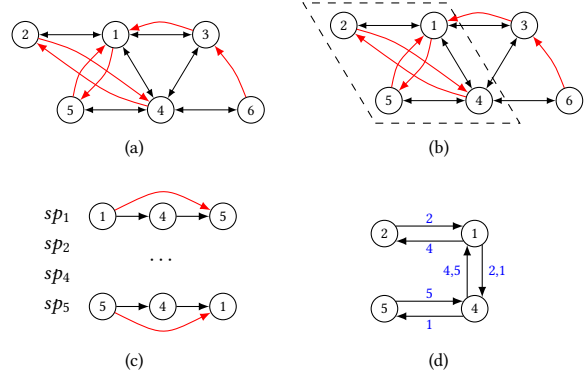
\begin{figure}[ht]
\centering
\subfigure[]{
\centering
\begin{minipage}[t]{0.45\linewidth}
\centering
\begin{tikzpicture}[scale=0.2, sibling distance=0em,
  every node/.style = {scale=0.65, shape=circle, draw, align=center},
    outline/.style={draw=#1,thick,fill=#1!100}]
  every draw/.style = {scale=2}
  \node[] (node5) at (3,-3) {5};
  \node[] (node4) at (9,-3) {4};
  \node[] (node6) at (15,-3) {6};
  \node[] (node2) at (0,2) {2};
  \node[] (node1) at (6,2) {1};
  \node[] (node3) at (12,2) {3};
  \draw[latex-latex,very thin] (node2)--(node1);
  \draw[latex-latex,very thin] (node1)--(node3);
  \draw[latex-latex,very thin] (node1)--(node4);
  \draw[latex-latex,very thin] (node3)--(node4);
  \draw[latex-latex,very thin] (node4)--(node5);
  \draw[latex-latex,very thin] (node6)--(node4);

  \draw[-latex,red] (node2).. controls(3.5,1) ..(node4);
  \draw[-latex,red] (node4).. controls(5.5,-2) ..(node2);
  \draw[-latex,red] (node1).. controls(5.5,-0.5) ..(node5);
  \draw[-latex,red] (node5).. controls(3.5,-0.5) ..(node1);
  \draw[-latex,red] (node3).. controls(9,3) ..(node1);
  \draw[-latex,red] (node6).. controls(14,-0.5) ..(node3);

\end{tikzpicture}

\end{minipage}
}
\subfigure[]{
\centering
\begin{minipage}[t]{0.45\linewidth}
\centering
\begin{tikzpicture}[scale=0.2, sibling distance=0em,
  every node/.style = {scale=0.65, shape=circle, draw, align=center},
    outline/.style={draw=#1,thick,fill=#1!100}]
  every draw/.style = {scale=2}
  \node[] (node5) at (3,-3) {5};
  \node[] (node4) at (9,-3) {4};
  \node[] (node6) at (15,-3) {6};
  \node[] (node2) at (0,2) {2};
  \node[] (node1) at (6,2) {1};
  \node[] (node3) at (12,2) {3};
  \draw[latex-latex,very thin] (node2)--(node1);
  \draw[latex-latex,very thin] (node1)--(node3);
  \draw[latex-latex,very thin] (node1)--(node4);
  \draw[latex-latex,very thin] (node3)--(node4);
  \draw[latex-latex,very thin] (node4)--(node5);
  \draw[latex-latex,very thin] (node6)--(node4);

  \draw[-latex,red] (node2).. controls(3.5,1) ..(node4);
  \draw[-latex,red] (node4).. controls(5.5,-2) ..(node2);
  \draw[-latex,red] (node1).. controls(5.5,-0.5) ..(node5);
  \draw[-latex,red] (node5).. controls(3.5,-0.5) ..(node1);
  \draw[-latex,red] (node3).. controls(9,3) ..(node1);
  \draw[-latex,red] (node6).. controls(14,-0.5) ..(node3);

  \draw[dashed,color=black] (2.1,-4.5) -- (11.7,-4.5) -- (6.9,3.5) -- (-2.7,3.5) -- (2.1,-4.5);

\end{tikzpicture}

\end{minipage}
}
\medskip
\subfigure[]{
\centering
\begin{minipage}[t]{0.45\linewidth}
\centering
\begin{tikzpicture}[scale=0.2, sibling distance=0em,
  every node/.style = {scale=0.65, shape=circle, draw, align=center},
    outline/.style={draw=#1,thick,fill=#1!100}]
    \node[rectangle,draw=none,font=\fontsize{12}{6}\selectfont] (node01) at (0,3) {$sp_1$};
    \node[] (node11) at (3,3) {1};
    \node[] (node12) at (7,3) {4};
    \node[] (node14) at (11,3) {5};
    \draw[latex-] (node12)--(node11);
    \draw[latex-] (node14)--(node12);
    \draw[-latex,red] (node11).. controls(7,5) ..(node14);
    \node[rectangle,draw=none,font=\fontsize{12}{6}\selectfont] (node02) at (0,1) {$sp_2$};
    \node[rectangle,draw=none,font=\fontsize{12}{6}\selectfont] (node00) at (7,0) {$\cdots$};
    \node[rectangle,draw=none,font=\fontsize{12}{6}\selectfont] (node04) at (0,-1) {$sp_4$};
    \node[rectangle,draw=none,font=\fontsize{12}{6}\selectfont] (node05) at (0,-3) {$sp_5$};
    \node[] (node51) at (3,-3) {5};
    \node[] (node52) at (7,-3) {4};
    \node[] (node54) at (11,-3) {1};
    \draw[latex-] (node52)--(node51);
    \draw[latex-] (node54)--(node52);
    \draw[-latex,red] (node51).. controls(7,-5) ..(node54);
\end{tikzpicture}
\end{minipage}
}
\subfigure[]{
\centering
\begin{minipage}[t]{0.45\linewidth}
\centering
\begin{tikzpicture}[scale=0.2, sibling distance=0em,
  every node/.style = {scale=0.65, shape=circle, draw, align=center},
    outline/.style={draw=#1,thick,fill=#1!100}]
  every draw/.style = {scale=2}
  \node[] (node5) at (3,-3) {5};
  \node[] (node4) at (9,-3) {4};
  \node[] (node2) at (3,2) {2};
  \node[] (node1) at (9,2) {1};

  
  \draw[-latex] (8.5,-2.25) -- (8.5,1.25);
  \draw[-latex] (9.5,1.25) -- (9.5,-2.25);
  \draw[-latex] (3.75,-2.5) -- (8.25,-2.5);
  \draw[-latex] (8.25,-3.5) -- (3.75,-3.5);
  \draw[-latex] (3.75,2.5) -- (8.25,2.5);
  \draw[-latex] (8.25,1.5) -- (3.75,1.5);
  
  \node[draw=none,text=blue,font=\fontsize{9}{6}\selectfont] (node14) at (7.5,-0.5) {4,5};
  \node[draw=none,text=blue,font=\fontsize{9}{6}\selectfont] (node41) at (10.5,-0.5) {2,1};
  \node[draw=none,text=blue,font=\fontsize{9}{6}\selectfont] (node21) at (6,3) {2};
  \node[draw=none,text=blue,font=\fontsize{9}{6}\selectfont] (node12) at (6,1) {4};
  \node[draw=none,text=blue,font=\fontsize{9}{6}\selectfont] (node45) at (6,-2) {5};
  \node[draw=none,text=blue,font=\fontsize{9}{6}\selectfont] (node54) at (6,-4) {1};

\end{tikzpicture}

\end{minipage}
}
\caption{Path Detection process for cycle $C_5=\{1,5\}$. The double arrows represent agents' connections. The red arrow is agents' pointing in the favorite pointing graph. The dashed box shows the minimum connected component $T$.}
\label{fig:path_detection_example}
\Description{A mini-example for Path Detection process.}
\end{figure}





We use a brief example to show how path detection works. As shown in Figure~\ref{fig:path_detection_example}, agents' social connections (edges in $G(\theta')$) are denoted by double arrows, and the red arrows represent edges in the favorite pointing graph $F(\theta')$. Take cycle $C_5=\{1,5\}$ as an example, the minimum connected component for $C_5$ is $T=\{1,2,4,5\}$. Agent 4 is involved to help agent 1 and 5 connect to each other, and since $\langle4,2\rangle\in F(\theta')$, agent 2 should also be involved in $T$ (as in Path Detection step a). Then we detect the shortest path set $SP$ for subgraph $G_T$, the sorted path set is shown in Figure~\ref{fig:path_detection_example}(c). Finally, follow the instruction in Path Detection step d.ii, we mark the subgraph $G_T$ to represents which directed edge by which agent. 

In cycles like $C_4=\{1,2,3,4\}$ or $C_5=\{1,5\}$, not all agents have an exclusive path to their pointing, so some have to switch their pointing in $F(\theta')$ to a less preferred one. We define a next favorite function to adjust $F(\theta')$ based on agents' preferences.

\begin{definition}
Given a reported type profile $\theta'$, the qualified agent set is $Q(\theta')$. Suppose agent $i$ points to agent $j$ in the favorite pointing graph $F(\theta')$, we have $\succ_i^{next}(F(\theta'))=j'$ if $j'$ owns $i$'s favorite item in $Q(\theta')\setminus\{k\in Q(\theta')|k\succeq_i j\}$.
\end{definition}

Now we are ready to introduce our mechanism. 

\begin{framed}
\noindent\textbf{Connected Trading Cycles (CTC)}

\noindent Given a reported type profile $\theta'$, construct the reported social network $G(\theta')$ and favorite pointing graph $F(\theta')$. Initiate a settled agent set $V=\emptyset$. Execute the following steps until $V=N$.

\begin{enumerate}
    \item [a.] Find the agent $p_1 \in N \setminus V$ with the minimum order $\mathcal{O}^{-1}(p_1)$. Start from agent $p_1$, detect a node sequence $P_m = \{p_1,\dots,p_l,\dots,p_m\}$ in $F(\theta')$ such that $\forall i\in \{1,\dots,m-1\}$, $p_i$ points to $p_{i+1}$, and $\{p_l,\dots,p_m\}$ is a cycle $C_m$.

    \item [b.] Run \textbf{Path Detection} algorithm with input $C_m$, $\theta'$, $G(\theta')$ and $F(\theta')$. Get an agent set $T$ and marked subgraph $G_T$ as output.

    \item [c.] Find every agent $i\in T$ such that no path in $G_T$ from $i$ to her pointing is exclusively marked by $i$, and $i$'s outgoing edges are marked only by $i$. Denote the set of such agents as $S\subset T$.
    \begin{enumerate}
        \item [i.] If $T=\emptyset$, find the last agent $i$ on $C_m$ who cannot connect to her pointing in $G(\theta')$, let $i$ switch her pointing to $\succ_{i}^{next}(F(\theta'))$.
        \item [ii.] If $S=\emptyset$ and $C_m=T$, add all agents in $T$ into the settled agent set $V$. While an agent $i\in N\setminus V$ points to an agent in $V$, let $i$ point to $\succ_i^{next}(F(\theta'))$.

        \item [iii.] If $S=\emptyset$ and $C_m\subset T$, find the agent $j\in T\setminus C_m$ with minimum order  $\mathcal{O}^{-1}(j)$. Start from agent $j$'s pointing, find the last agent $i\in C_m$ such that $i$'s path to its pointing passes at least one agent in $T\setminus C_m$. If no such $i$, start from agent $j$'s pointing, find the last agent $i\in (T\setminus C_m)$ such that $i$'s path to its pointing passes at least one agent in $C_m$. Let $i$ point to $\succ_i^{next}(F(\theta'))$.

        \item [iv.] If $(T\setminus C_m) \cap S\neq \emptyset$, find the agent $i\in (T\setminus C_m) \cap S$ with minimum order $\mathcal{O}^{-1}(i)$. Let $i$ point to $\succ_{i}^{next}(F(\theta'))$.
        \item [v.] If $(T\setminus C_m) \cap S= \emptyset$, find the agent $i\in S$ with the minimum order $\mathcal{O}^{-1}(i)$ such that $i$'s path to its pointing covers another agent's path to her pointing. If no such $i$, find agent $i\in S$ with the minimum order $\mathcal{O}^{-1}(i)$. Let $i$ point to $\succ_{i}^{next}(F(\theta'))$.
    \end{enumerate}
\end{enumerate}

\end{framed}

Our mechanism first detects cycles in the favorite pointing graph because these cycles can ensure optimality. Thus, agents will not misreport their preferences to get a better allocation (also why TTC is IC in the traditional setting). However, in the network setting, trading cycles formed by agents who cannot connect to each other are fragile. The reason is that others can misreport neighbor set to break the cycle (for instance, disqualify some agents on the cycle). Hence, we construct a minimum connected component $T$ for each cycle where they can stay connected. 

Recall the example in Figure~\ref{fig:figure_cycle}, we can see that if every agent on the cycle has an exclusive path to her pointing (e.g. 1-3-5), the cycle can get traded. 
For cycles that do not meet this condition, we identify a special agent set $S$. each agent in $S$ does not have an exclusive path to her pointing, and her outgoing edges are only marked by herself. These agents cannot misreport to avoid path overlap because the overlap lies outside their outgoing edges (which is irrelevant to their report). Hence, we let agents in $S$ switch their pointing to a less preferred one.

For cases where $S=\emptyset$, if $T$ consists of only agents on the cycle, they can get traded (CTC step c-ii). Otherwise, we start from an agent $i\in T$ who is not on the cycle and find the last agent on the cycle who does not have a path to her pointing only by connections on the cycle (CTC step c-iii). This is because every agent in $T$ now has an exclusive path to her pointing, and the cycle has to rely on agents not on the cycle to stay connected. Thus, agents on the cycle have to switch preferences first, and we start from the last agent to maintain the completeness of the cycle to the greatest extent.

For cases where $S\neq\emptyset$, if $S$ contains an agent not on the cycle, let her switch pointing to a less preferred one first (CTC step c-iv). This is because such agents do not provide any connection for the cycle and agents on the cycle have incentives to disqualify them to eliminate path overlaps. If $S$ consists only of agents on the cycle, we find the agent $i$ whose path covers another agent $j$'s path to her pointing and let $i$ switch her preference (CTC step c-v). This is because, to ensure everyone has an exclusive path to her pointing, $i$ and $j$ can never hold their pointing together. Agent $i$ should switch preference first because her path contains $j$'s outgoing edges ($j$'s report has impacts on $i$).

\begin{figure}[ht]
\centering
\subfigure[]{
\centering
\begin{minipage}[t]{0.45\linewidth}
\centering
\begin{tikzpicture}[scale=0.2, sibling distance=0em,
  every node/.style = {scale=0.65, shape=circle, draw, align=center},
    outline/.style={draw=#1,thick,fill=#1!100}]
  every draw/.style = {scale=2}
  \node[] (node5) at (3,-3) {5};
  \node[] (node4) at (9,-3) {4};
  \node[] (node6) at (15,-3) {6};
  \node[] (node2) at (0,2) {2};
  \node[line width=1pt] (node1) at (6,2) {1};
  \node[] (node3) at (12,2) {3};
  \draw[latex-latex,very thin] (node2)--(node1);
  \draw[latex-latex,very thin] (node1)--(node3);
  \draw[latex-latex,very thin] (node1)--(node4);
  \draw[latex-latex,very thin] (node3)--(node4);
  \draw[latex-latex,very thin] (node4)--(node5);
  \draw[latex-latex,very thin] (node6)--(node4);

  \draw[-latex,red] (node2).. controls(3.5,1) ..(node4);
  \draw[-latex,red] (node4).. controls(5.5,-2) ..(node2);
  \draw[-latex,red] (node1).. controls(5.5,-0.5) ..(node5);
  \draw[-latex,red] (node5).. controls(3.5,-0.5) ..(node1);
  \draw[-latex,red] (node3).. controls(9,3) ..(node1);
  \draw[-latex,red] (node6).. controls(14,-0.5) ..(node3);

\end{tikzpicture}

\end{minipage}
}
\subfigure[]{
\centering
\begin{minipage}[t]{0.45\linewidth}
\centering
\begin{tikzpicture}[scale=0.2, sibling distance=0em,
  every node/.style = {scale=0.65, shape=circle, draw, align=center},
    outline/.style={draw=#1,thick,fill=#1!100}]
  every draw/.style = {scale=2}
  \node[] (node5) at (3,-3) {5};
  \node[] (node4) at (9,-3) {4};
  \node[] (node6) at (15,-3) {6};
  \node[] (node2) at (0,2) {2};
  \node[line width=1pt] (node1) at (6,2) {1};
  \node[] (node3) at (12,2) {3};
  \draw[latex-latex,very thin] (node2)--(node1);
  \draw[latex-latex,very thin] (node1)--(node3);
  \draw[latex-latex,very thin] (node1)--(node4);
  \draw[latex-latex,very thin] (node3)--(node4);
  \draw[latex-latex,very thin] (node4)--(node5);
  \draw[latex-latex,very thin] (node6)--(node4);

  \draw[-latex,red] (node2).. controls(3.5,1) ..(node4);
  \draw[-latex,red] (node4).. controls(5.5,-2) ..(node2);
  \draw[-latex,red] (node1).. controls(5.5,-0.5) ..(node5);
  \draw[-latex,red] (node5).. controls(3.5,-0.5) ..(node1);
  \draw[-latex,red] (node3).. controls(9,3) ..(node1);
  \draw[-latex,red] (node6).. controls(14,-0.5) ..(node3);

  \draw[dashed,color=black] (2.1,-4.5) -- (11.7,-4.5) -- (6.9,3.5) -- (-2.7,3.5) -- (2.1,-4.5);

\end{tikzpicture}

\end{minipage}
}
\medskip
\subfigure[]{
\centering
\begin{minipage}[t]{0.45\linewidth}
\centering
\begin{tikzpicture}[scale=0.2, sibling distance=0em,
  every node/.style = {scale=0.65, shape=circle, draw, align=center},
    outline/.style={draw=#1,thick,fill=#1!100}]
  every draw/.style = {scale=2}
  \node[] (node5) at (3,-3) {5};
  \node[] (node4) at (9,-3) {4};
  \node[] (node6) at (15,-3) {6};
  \node[] (node2) at (0,2) {2};
  \node[line width=1pt] (node1) at (6,2) {1};
  \node[] (node3) at (12,2) {3};
  \draw[latex-latex,very thin] (node2)--(node1);
  \draw[latex-latex,very thin] (node1)--(node3);
  \draw[latex-latex,very thin] (node1)--(node4);
  \draw[latex-latex,very thin] (node3)--(node4);
  \draw[latex-latex,very thin] (node4)--(node5);
  \draw[latex-latex,very thin] (node6)--(node4);

  \draw[-latex,red,dashed] (node2).. controls(3,3) ..(node1);
  \draw[-latex,red] (node4).. controls(5.5,-2) ..(node2);
  \draw[-latex,red] (node1).. controls(5.5,-0.5) ..(node5);
  \draw[-latex,red] (node5).. controls(3.5,-0.5) ..(node1);
  \draw[-latex,red] (node3).. controls(9,3) ..(node1);
  \draw[-latex,red] (node6).. controls(14,-0.5) ..(node3);

  \draw[dashed,color=black] (2.1,-4.5) -- (11.7,-4.5) -- (6.9,3.5) -- (-2.7,3.5) -- (2.1,-4.5);

\end{tikzpicture}

\end{minipage}
}
\subfigure[]{
\centering
\begin{minipage}[t]{0.45\linewidth}
\centering
\begin{tikzpicture}[scale=0.2, sibling distance=0em,
  every node/.style = {scale=0.65, shape=circle, draw, align=center},
    outline/.style={draw=#1,thick,fill=#1!100}]
  every draw/.style = {scale=2}
  \node[] (node5) at (3,-3) {5};
  \node[] (node4) at (9,-3) {4};
  \node[] (node6) at (15,-3) {6};
  \node[] (node2) at (0,2) {2};
  \node[line width=1pt] (node1) at (6,2) {1};
  \node[] (node3) at (12,2) {3};
  \draw[latex-latex,very thin] (node2)--(node1);
  \draw[latex-latex,very thin] (node1)--(node3);
  \draw[latex-latex,very thin] (node1)--(node4);
  \draw[latex-latex,very thin] (node3)--(node4);
  \draw[latex-latex,very thin] (node4)--(node5);
  \draw[latex-latex,very thin] (node6)--(node4);

  \draw[-latex,red,dashed] (node2).. controls(3,3) ..(node1);
  \draw[-latex,red] (node4).. controls(5.5,-2) ..(node2);
  \draw[-latex,red] (node1).. controls(5.5,-0.5) ..(node5);
  \draw[-latex,red,dashed] (node5).. controls(7.5,-0.5) ..(node3);
  \draw[-latex,red] (node3).. controls(9,3) ..(node1);
  \draw[-latex,red] (node6).. controls(14,-0.5) ..(node3);

  \draw[dashed,color=black] (0.3,-4.5) -- (11.7,-4.5) -- (14.7,3.5) -- (-2.7,3.5) -- (0.3,-4.5);


\end{tikzpicture}

\end{minipage}
}
\medskip
\subfigure[]{
\centering
\begin{minipage}[t]{0.45\linewidth}
\centering
\begin{tikzpicture}[scale=0.2, sibling distance=0em,
  every node/.style = {scale=0.65, shape=circle, draw, align=center},
    outline/.style={draw=#1,thick,fill=#1!100}]
  every draw/.style = {scale=2}
  \node[] (node5) at (3,-3) {5};
  \node[] (node4) at (9,-3) {4};
  \node[] (node6) at (15,-3) {6};
  \node[] (node2) at (0,2) {2};
  \node[line width=1pt] (node1) at (6,2) {1};
  \node[] (node3) at (12,2) {3};
  \draw[latex-latex,very thin] (node2)--(node1);
  \draw[latex-latex,very thin] (node1)--(node3);
  \draw[latex-latex,very thin] (node1)--(node4);
  \draw[latex-latex,very thin] (node3)--(node4);
  \draw[latex-latex,very thin] (node4)--(node5);
  \draw[latex-latex,very thin] (node6)--(node4);

  \draw[-latex,red,dashed] (node2).. controls(3,3) ..(node1);
  \draw[-latex,red] (node4).. controls(5.5,-2) ..(node2);
  \draw[-latex,red] (node1).. controls(5.5,-0.5) ..(node5);
  \draw[-latex,red,dashed] (node5).. controls(6,-4.3) ..(node4);
  \draw[-latex,red] (node3).. controls(9,3) ..(node1);
  \draw[-latex,red] (node6).. controls(14,-0.5) ..(node3);

  \draw[dashed,color=black] (2.1,-4.5) -- (11.7,-4.5) -- (6.9,3.5) -- (-2.7,3.5) -- (2.1,-4.5);


\end{tikzpicture}

\end{minipage}
}
\subfigure[]{
\centering
\begin{minipage}[t]{0.45\linewidth}
\centering
\begin{tikzpicture}[scale=0.2, sibling distance=0em,
  every node/.style = {scale=0.65, shape=circle, draw, align=center},
    outline/.style={draw=#1,thick,fill=#1!100}]
  every draw/.style = {scale=2}
  \node[fill=black!10] (node5) at (3,-3) {5};
  \node[fill=black!10] (node4) at (9,-3) {4};
  \node[] (node6) at (15,-3) {6};
  \node[fill=black!10] (node2) at (0,2) {2};
  \node[fill=black!10,line width=1pt] (node1) at (6,2) {1};
  \node[] (node3) at (12,2) {3};
  \draw[latex-latex,very thin] (node2)--(node1);
  \draw[latex-latex,very thin] (node1)--(node3);
  \draw[latex-latex,very thin] (node1)--(node4);
  \draw[latex-latex,very thin] (node3)--(node4);
  \draw[latex-latex,very thin] (node4)--(node5);
  \draw[latex-latex,very thin] (node6)--(node4);
  \draw[-latex,red,dashed] (node2).. controls(3,3) ..(node1);
  \draw[-latex,red] (node4).. controls(5.5,-2) ..(node2);
  \draw[-latex,red] (node1).. controls(5.5,-0.5) ..(node5);
  \draw[-latex,red,dashed] (node5).. controls(6,-4.3) ..(node4);
  \draw[-latex,red,dashed] (node3).. controls(12,-0.5) ..(node6);
  \draw[-latex,red] (node6).. controls(14,-0.5) ..(node3);

\end{tikzpicture}

\end{minipage}
}
\caption{The double arrows represent agents' connections. The red arrow is agents' pointing in the favorite pointing graph. The dashed red arrow means the agent switch her pointing to a less favorite one. The dashed box shows the minimum connected component $T$.}
\label{fig:mini_example}
\Description{A mini-example for \textit{Connected Trading Cycles}, see the full version in Appendix A.}
\end{figure}
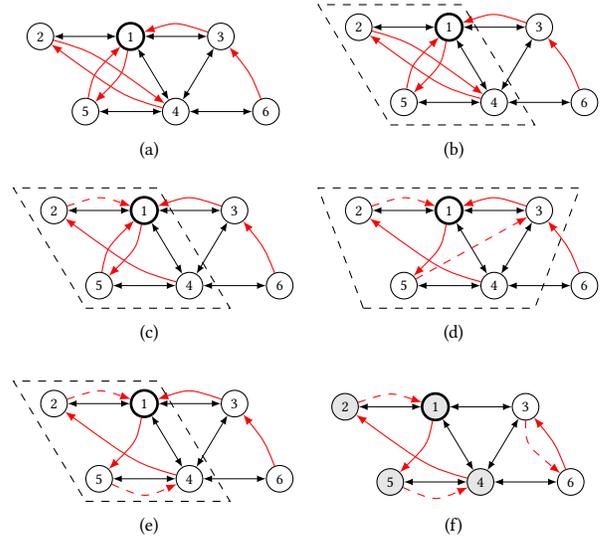

We present an example to illustrate the execution of CTC in Figure~\ref{fig:mini_example}. In (a), everyone points to her favorite item. Since agent 1 has the minimum order, we start from agent 1 and detect cycle $C_5=\{1,5\}$. Then, we construct the minimum connected component $T=\{1,2,4,5\}$ in (b). By the Path Detection algorithm, we can see that both agent 2 and agent 5 do not have an exclusive path to their pointing, so $S=\{2,5\}$. By CTC step c-iv, we let agent 2 switch her pointing to a less favorite one. In (c), we also start from agent 1 and detect the same cycle $C_5$, and the same $T$. Now, only agent 5 does not have an exclusive path to her pointing, so $S=\{5\}$. By CTC step c-v, we let agent 5 switch her pointing. In (d), we detect cycle $C_3=\{1,5,3\}$ and construct $T=\{1,2,3,4,5\}$. Everyone has an exclusive path to reach her pointing, so $S=\emptyset$. By CTC step c-iii, we start from agent 2's pointing and find that agent 5 is the last agent whose path to her pointing passes at least one agent in $T\setminus C_3$. So, we let agent 5 switch her pointing. Finally in (e), $C_2=\{1,5,4,2\}$ is a connected cycle and $S=\emptyset$. By CTC step c-ii, $C_2$ can get traded and agent 3 will switch her pointing to her next favorite one as shown in (f). We leave the detailed execution in Appendix A.

To our best knowledge, the Connected Trading Cycles is the first mechanism that satisfies optimal-cc.

\begin{theorem}
For any order $\mathcal{O}$, CTC is optimal-cc.
\end{theorem}
\begin{proof}
For any given type profile $\theta$, if the allocation $\pi(\theta)$ given by CTC violates optimal-cc, there exists another allocation $\pi'(\theta)$ such that every agent in a complete component $B$ in $G(\theta)$ is strictly better, while others receive the same allocation. Let $i,j$ be two agents in $B$, and $\pi'_i(\theta)=\pi_j(\theta)=h_k$. Because $B$ is a complete component in $G(\theta)$, $i$ can reach $j$ by her outgoing edge, which means $i\to j$ is an exclusive path for $i$. Since $\pi(\theta)$ is given by CTC, $j$ has an exclusive path to agent $k$, the owner of her allocation $\pi_j(\theta)$ (as in CTC step c.ii). Combining the two paths together, agent $i$ can reach $k$, the owner of $\pi'_i(\theta)$, by $i \to j \to k$. Hence, if $i$ points to $k$, CTC will not make agent $i$ switch her pointing to a less preferred agents. Therefore, there exists an exclusive path for every agent $i\in B$ to reach the owner of $\pi'_i(\theta)$, and CTC will allocate $\pi'(\theta)$. This means $\pi(\theta)=\pi'(\theta)$, which contradicts our assumption.
\end{proof}

\section{Conclusion}
In this paper, we redefine optimality and stability notions in social networks and show their tightness by proving impossibilities. We propose two novel matching protocols called Swap With Neighbors and Leave and Share respectively that satisfy IC, IR and stable-cc. We also prove the theoretical boundaries and propose a mechanism called Connected Trading Cycles to reach them. One future work is to study network settings with practical limits to find better optimality and stability notions. Another direction is to design the distributed execution for centralized matching mechanisms to make them more applicable.

\begin{acks}
We thanks the reviewers for their valuable comments on the original submission. This work was supported in part by Science and Technology Commission of Shanghai Municipality (No.22ZR1442200 and No.23010503000), and Shanghai Frontiers Science Center of Human-centered Artificial Intelligence (ShangHAI).
\end{acks}

\bibliographystyle{matching} 
\bibliography{matching}


\begin{thebibliography}{17}


\ifx \showCODEN    \undefined \def \showCODEN     #1{\unskip}     \fi
\ifx \showDOI      \undefined \def \showDOI       #1{#1}\fi
\ifx \showISBNx    \undefined \def \showISBNx     #1{\unskip}     \fi
\ifx \showISBNxiii \undefined \def \showISBNxiii  #1{\unskip}     \fi
\ifx \showISSN     \undefined \def \showISSN      #1{\unskip}     \fi
\ifx \showLCCN     \undefined \def \showLCCN      #1{\unskip}     \fi
\ifx \shownote     \undefined \def \shownote      #1{#1}          \fi
\ifx \showarticletitle \undefined \def \showarticletitle #1{#1}   \fi
\ifx \showURL      \undefined \def \showURL       {\relax}        \fi
\providecommand\bibfield[2]{#2}
\providecommand\bibinfo[2]{#2}
\providecommand\natexlab[1]{#1}
\providecommand\showeprint[2][]{arXiv:#2}

\bibitem[\protect\citeauthoryear{Abdulkadiro{\u{g}}lu and S{\"o}nmez}{Abdulkadiro{\u{g}}lu and S{\"o}nmez}{1999}]%
        {abdulkadirouglu1999house}
\bibfield{author}{\bibinfo{person}{Atila Abdulkadiro{\u{g}}lu} {and} \bibinfo{person}{Tayfun S{\"o}nmez}.} \bibinfo{year}{1999}\natexlab{}.
\newblock \showarticletitle{House allocation with existing tenants}.
\newblock \bibinfo{journal}{\emph{Journal of Economic Theory}} \bibinfo{volume}{88}, \bibinfo{number}{2} (\bibinfo{year}{1999}), \bibinfo{pages}{233--260}.
\newblock


\bibitem[\protect\citeauthoryear{Abdulkadiroglu and S{\"o}nmez}{Abdulkadiroglu and S{\"o}nmez}{2013}]%
        {abdulkadiroglu2013matching}
\bibfield{author}{\bibinfo{person}{Atila Abdulkadiroglu} {and} \bibinfo{person}{Tayfun S{\"o}nmez}.} \bibinfo{year}{2013}\natexlab{}.
\newblock \showarticletitle{Matching markets: Theory and practice}.
\newblock \bibinfo{journal}{\emph{Advances in Economics and Econometrics}}  \bibinfo{volume}{1} (\bibinfo{year}{2013}), \bibinfo{pages}{3--47}.
\newblock


\bibitem[\protect\citeauthoryear{Abraham, Cechl{\'a}rov{\'a}, Manlove, and Mehlhorn}{Abraham et~al\mbox{.}}{2005}]%
        {abraham2004pareto}
\bibfield{author}{\bibinfo{person}{David~J. Abraham}, \bibinfo{person}{Katar{\'i}na Cechl{\'a}rov{\'a}}, \bibinfo{person}{David~F. Manlove}, {and} \bibinfo{person}{Kurt Mehlhorn}.} \bibinfo{year}{2005}\natexlab{}.
\newblock \showarticletitle{Pareto Optimality in House Allocation Problems}. In \bibinfo{booktitle}{\emph{Algorithms and Computation}}, \bibfield{editor}{\bibinfo{person}{Rudolf Fleischer} {and} \bibinfo{person}{Gerhard Trippen}} (Eds.). \bibinfo{publisher}{Springer Berlin Heidelberg}, \bibinfo{address}{Berlin, Heidelberg}, \bibinfo{pages}{3--15}.
\newblock
\showISBNx{978-3-540-30551-4}


\bibitem[\protect\citeauthoryear{Brandt and Wilczynski}{Brandt and Wilczynski}{2019}]%
        {brandt2019convergence}
\bibfield{author}{\bibinfo{person}{Felix Brandt} {and} \bibinfo{person}{Ana{\"e}lle Wilczynski}.} \bibinfo{year}{2019}\natexlab{}.
\newblock \showarticletitle{On the Convergence of Swap Dynamics to Pareto-Optimal Matchings}. In \bibinfo{booktitle}{\emph{Web and Internet Economics}}, \bibfield{editor}{\bibinfo{person}{Ioannis Caragiannis}, \bibinfo{person}{Vahab Mirrokni}, {and} \bibinfo{person}{Evdokia Nikolova}} (Eds.). \bibinfo{publisher}{Springer International Publishing}, \bibinfo{address}{Cham}, \bibinfo{pages}{100--113}.
\newblock
\showISBNx{978-3-030-35389-6}


\bibitem[\protect\citeauthoryear{Cho, Todo, and Yokoo}{Cho et~al\mbox{.}}{2022}]%
        {ijcai2022-27}
\bibfield{author}{\bibinfo{person}{Sung-Ho Cho}, \bibinfo{person}{Taiki Todo}, {and} \bibinfo{person}{Makoto Yokoo}.} \bibinfo{year}{2022}\natexlab{}.
\newblock \showarticletitle{Two-Sided Matching over Social Networks}. In \bibinfo{booktitle}{\emph{Proceedings of the Thirty-First International Joint Conference on Artificial Intelligence, {IJCAI-22}}}, \bibfield{editor}{\bibinfo{person}{Lud~De Raedt}} (Ed.). \bibinfo{publisher}{International Joint Conferences on Artificial Intelligence Organization}, \bibinfo{pages}{186--193}.
\newblock
\urldef\tempurl%
\url{https://doi.org/10.24963/ijcai.2022/27}
\showDOI{\tempurl}
\newblock
\shownote{Main Track.}


\bibitem[\protect\citeauthoryear{Fleischer and Wang}{Fleischer and Wang}{2008}]%
        {fleischer2008dynamic}
\bibfield{author}{\bibinfo{person}{Rudolf Fleischer} {and} \bibinfo{person}{Yihui Wang}.} \bibinfo{year}{2008}\natexlab{}.
\newblock \showarticletitle{Dynamic Pareto Optimal Matching}. In \bibinfo{booktitle}{\emph{Proceedings of the 2008 International Symposium on Information Science and Engieering - Volume 02}} \emph{(\bibinfo{series}{ISISE '08})}. \bibinfo{publisher}{IEEE Computer Society}, \bibinfo{address}{USA}, \bibinfo{pages}{797–802}.
\newblock
\showISBNx{9780769534947}
\urldef\tempurl%
\url{https://doi.org/10.1109/ISISE.2008.237}
\showDOI{\tempurl}


\bibitem[\protect\citeauthoryear{Gourv\`{e}s, Lesca, and Wilczynski}{Gourv\`{e}s et~al\mbox{.}}{2017}]%
        {DBLP:conf/ijcai/GourvesLW17}
\bibfield{author}{\bibinfo{person}{Laurent Gourv\`{e}s}, \bibinfo{person}{Julien Lesca}, {and} \bibinfo{person}{Ana\"{e}lle Wilczynski}.} \bibinfo{year}{2017}\natexlab{}.
\newblock \showarticletitle{Object Allocation via Swaps along a Social Network}. In \bibinfo{booktitle}{\emph{Proceedings of the 26th International Joint Conference on Artificial Intelligence}} (Melbourne, Australia) \emph{(\bibinfo{series}{IJCAI'17})}. \bibinfo{publisher}{AAAI Press}, \bibinfo{pages}{213–219}.
\newblock
\showISBNx{9780999241103}


\bibitem[\protect\citeauthoryear{Hoefer}{Hoefer}{2013}]%
        {HOEFER201320}
\bibfield{author}{\bibinfo{person}{Martin Hoefer}.} \bibinfo{year}{2013}\natexlab{}.
\newblock \showarticletitle{Local matching dynamics in social networks}.
\newblock \bibinfo{journal}{\emph{Information and Computation}}  \bibinfo{volume}{222} (\bibinfo{year}{2013}), \bibinfo{pages}{20--35}.
\newblock
\showISSN{0890-5401}
\urldef\tempurl%
\url{https://doi.org/10.1016/j.ic.2012.10.005}
\showDOI{\tempurl}
\newblock
\shownote{38th International Colloquium on Automata, Languages and Programming (ICALP 2011).}


\bibitem[\protect\citeauthoryear{Kawasaki, Wada, Todo, and Yokoo}{Kawasaki et~al\mbox{.}}{2021}]%
        {DBLP:conf/atal/KawasakiWTY21}
\bibfield{author}{\bibinfo{person}{Takehiro Kawasaki}, \bibinfo{person}{Ryoji Wada}, \bibinfo{person}{Taiki Todo}, {and} \bibinfo{person}{Makoto Yokoo}.} \bibinfo{year}{2021}\natexlab{}.
\newblock \showarticletitle{Mechanism Design for Housing Markets over Social Networks}. In \bibinfo{booktitle}{\emph{Proceedings of the 20th International Conference on Autonomous Agents and MultiAgent Systems}} (Virtual Event, United Kingdom) \emph{(\bibinfo{series}{AAMAS '21})}. \bibinfo{publisher}{International Foundation for Autonomous Agents and Multiagent Systems}, \bibinfo{address}{Richland, SC}, \bibinfo{pages}{692–700}.
\newblock
\showISBNx{9781450383073}


\bibitem[\protect\citeauthoryear{Kumar, Manocha, and Ortega}{Kumar et~al\mbox{.}}{2022}]%
        {kumar2022integration}
\bibfield{author}{\bibinfo{person}{Rajnish Kumar}, \bibinfo{person}{Kriti Manocha}, {and} \bibinfo{person}{Josu{\'e} Ortega}.} \bibinfo{year}{2022}\natexlab{}.
\newblock \showarticletitle{On the integration of Shapley--Scarf markets}.
\newblock \bibinfo{journal}{\emph{Journal of Mathematical Economics}}  \bibinfo{volume}{100} (\bibinfo{year}{2022}), \bibinfo{pages}{102637}.
\newblock


\bibitem[\protect\citeauthoryear{Li, Hao, Gao, and Zhao}{Li et~al\mbox{.}}{2022}]%
        {li2022diffusion}
\bibfield{author}{\bibinfo{person}{Bin Li}, \bibinfo{person}{Dong Hao}, \bibinfo{person}{Hui Gao}, {and} \bibinfo{person}{Dengji Zhao}.} \bibinfo{year}{2022}\natexlab{}.
\newblock \showarticletitle{Diffusion auction design}.
\newblock \bibinfo{journal}{\emph{Artificial Intelligence}}  \bibinfo{volume}{303} (\bibinfo{year}{2022}), \bibinfo{pages}{103631}.
\newblock


\bibitem[\protect\citeauthoryear{Li, Hao, Zhao, and Zhou}{Li et~al\mbox{.}}{2017}]%
        {li2017mechanism}
\bibfield{author}{\bibinfo{person}{Bin Li}, \bibinfo{person}{Dong Hao}, \bibinfo{person}{Dengji Zhao}, {and} \bibinfo{person}{Tao Zhou}.} \bibinfo{year}{2017}\natexlab{}.
\newblock \showarticletitle{Mechanism Design in Social Networks}. In \bibinfo{booktitle}{\emph{Proceedings of the Thirty-First AAAI Conference on Artificial Intelligence}} (San Francisco, California, USA) \emph{(\bibinfo{series}{AAAI'17})}. \bibinfo{publisher}{AAAI Press}, \bibinfo{pages}{586–592}.
\newblock


\bibitem[\protect\citeauthoryear{Ma}{Ma}{1994}]%
        {ma1994strategy}
\bibfield{author}{\bibinfo{person}{Jinpeng Ma}.} \bibinfo{year}{1994}\natexlab{}.
\newblock \showarticletitle{Strategy-proofness and the strict core in a market with indivisibilities}.
\newblock \bibinfo{journal}{\emph{International Journal of Game Theory}} \bibinfo{volume}{23}, \bibinfo{number}{1} (\bibinfo{year}{1994}), \bibinfo{pages}{75--83}.
\newblock


\bibitem[\protect\citeauthoryear{Shapley and Scarf}{Shapley and Scarf}{1974}]%
        {shapley1974cores}
\bibfield{author}{\bibinfo{person}{Lloyd Shapley} {and} \bibinfo{person}{Herbert Scarf}.} \bibinfo{year}{1974}\natexlab{}.
\newblock \showarticletitle{On cores and indivisibility}.
\newblock \bibinfo{journal}{\emph{Journal of mathematical economics}} \bibinfo{volume}{1}, \bibinfo{number}{1} (\bibinfo{year}{1974}), \bibinfo{pages}{23--37}.
\newblock


\bibitem[\protect\citeauthoryear{You, Dierks, Todo, Li, and Yokoo}{You et~al\mbox{.}}{2022}]%
        {tenants_social_network}
\bibfield{author}{\bibinfo{person}{Bo You}, \bibinfo{person}{Ludwig Dierks}, \bibinfo{person}{Taiki Todo}, \bibinfo{person}{Minming Li}, {and} \bibinfo{person}{Makoto Yokoo}.} \bibinfo{year}{2022}\natexlab{}.
\newblock \showarticletitle{Strategy-Proof House Allocation with Existing Tenants over Social Networks}. In \bibinfo{booktitle}{\emph{Proceedings of the 21st International Conference on Autonomous Agents and Multiagent Systems}} (Virtual Event, New Zealand) \emph{(\bibinfo{series}{AAMAS '22})}. \bibinfo{publisher}{International Foundation for Autonomous Agents and Multiagent Systems}, \bibinfo{address}{Richland, SC}, \bibinfo{pages}{1446–1454}.
\newblock
\showISBNx{9781450392136}


\bibitem[\protect\citeauthoryear{Zhang and Zhao}{Zhang and Zhao}{2022}]%
        {zhang2022incentives}
\bibfield{author}{\bibinfo{person}{Yao Zhang} {and} \bibinfo{person}{Dengji Zhao}.} \bibinfo{year}{2022}\natexlab{}.
\newblock \showarticletitle{Incentives to Invite Others to Form Larger Coalitions}. In \bibinfo{booktitle}{\emph{Proceedings of the 21st International Conference on Autonomous Agents and Multiagent Systems}} (Virtual Event, New Zealand) \emph{(\bibinfo{series}{AAMAS '22})}. \bibinfo{publisher}{International Foundation for Autonomous Agents and Multiagent Systems}, \bibinfo{address}{Richland, SC}, \bibinfo{pages}{1509–1517}.
\newblock
\showISBNx{9781450392136}


\bibitem[\protect\citeauthoryear{Zhao}{Zhao}{2021}]%
        {DBLP:conf/atal/Zhao21}
\bibfield{author}{\bibinfo{person}{Dengji Zhao}.} \bibinfo{year}{2021}\natexlab{}.
\newblock \showarticletitle{Mechanism Design Powered by Social Interactions}. In \bibinfo{booktitle}{\emph{Proceedings of the 20th International Conference on Autonomous Agents and MultiAgent Systems}} (Virtual Event, United Kingdom) \emph{(\bibinfo{series}{AAMAS '21})}. \bibinfo{publisher}{International Foundation for Autonomous Agents and Multiagent Systems}, \bibinfo{address}{Richland, SC}, \bibinfo{pages}{63–67}.
\newblock
\showISBNx{9781450383073}


\end{thebibliography}

\newpage
\appendix
\section{A Detailed Example for CTC}
\begin{figure}[ht]
\centering
\subfigure[]{
\centering
\begin{minipage}[t]{0.45\linewidth}
\centering

\begin{tikzpicture}[scale=0.2, sibling distance=0em,
  every node/.style = {scale=0.65, shape=circle, draw, align=center},
    outline/.style={draw=#1,thick,fill=#1!100}]
  every draw/.style = {scale=2}
  \node[] (node5) at (3,-3) {5};
  \node[] (node4) at (9,-3) {4};
  \node[] (node6) at (15,-3) {6};
  \node[] (node2) at (0,2) {2};
  \node[] (node1) at (6,2) {1};
  \node[] (node3) at (12,2) {3};
  \draw[latex-latex,very thin] (node2)--(node1);
  \draw[latex-latex,very thin] (node1)--(node3);
  \draw[latex-latex,very thin] (node1)--(node4);
  \draw[latex-latex,very thin] (node3)--(node4);
  \draw[latex-latex,very thin] (node4)--(node5);
  \draw[latex-latex,very thin] (node6)--(node4);
  
\end{tikzpicture}
\end{minipage}
}
\subfigure[]{
\centering
\begin{minipage}[t]{0.45\linewidth}
\centering

\begin{tikzpicture}[scale=0.2, sibling distance=0em,
  every node/.style = {scale=0.65, shape=circle, draw, align=center},
    outline/.style={draw=#1,thick,fill=#1!100}]
  every draw/.style = {scale=2}
  \node[] (node5) at (3,-3) {5};
  \node[] (node4) at (9,-3) {4};
  \node[] (node6) at (15,-3) {6};
  \node[] (node2) at (0,2) {2};
  \node[] (node1) at (6,2) {1};
  \node[] (node3) at (12,2) {3};
  \draw[-latex,red] (node2).. controls(3.5,1) ..(node4);
  \draw[-latex,red] (node4).. controls(5.5,-2) ..(node2);
  \draw[-latex,red] (node1).. controls(5.5,-0.5) ..(node5);
  \draw[-latex,red] (node5).. controls(3.5,-0.5) ..(node1);
  \draw[-latex,red] (node3).. controls(9,3) ..(node1);
  \draw[-latex,red] (node6).. controls(14,-0.5) ..(node3);
\end{tikzpicture}

\end{minipage}
}
\subfigure[]{
\centering
\begin{minipage}[ht]{0.65\linewidth}
\centering
\small
\begin{tabular}{rlll}\toprule
    \textit{i} & \textit{$\succ_i$} & \textit{$r_i$}  & \textit{$\pi_i$}\\ \midrule
    1 & $h_5 \succ h_1 \succ \cdots$ & 2,3,4 & $h_5$ \\
    2 & $h_4 \succ h_1 \succ h_2 \succ \cdots$ & 1 & $h_1$\\
    3 & $h_1 \succ h_6 \succ h_3 \succ \cdots$ & 1 & $h_3$\\
    4 & $h_2 \succ h_4 \succ \cdots $ & 1,3,5,6 & $h_2$\\
    5 & $h_1 \succ h_3 \succ h_4 \succ h_5 \succ \cdots$ & 4 & $h_4$\\
    6 & $h_3 \succ h_6 \succ \cdots$ & 4 & $h_6$\\ \bottomrule
\end{tabular}
\end{minipage}
}

\caption{Figure (a) represents the $G(\theta')$. Figure (b) is the favorite pointing graph $F(\theta')$. Table (c) presents the preference, neighbor set, and the allocation given by CTC for each agent.}
\label{fig:pref_table}
\Description{A detailed example for \textit{Connected Trading Cycles}, graph and profile table.}
\end{figure}
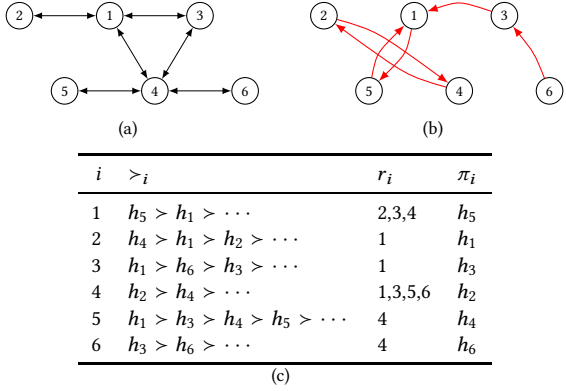

\begin{figure}[ht]
\centering
\subfigure[]{
\centering
\begin{minipage}[t]{0.45\linewidth}
\centering
\begin{tikzpicture}[scale=0.22, sibling distance=0em,
  every node/.style = {scale=0.65, shape=circle, draw, align=center},
    outline/.style={draw=#1,thick,fill=#1!100}]
  every draw/.style = {scale=2}
  \node[] (node5) at (3,-3) {5};
  \node[] (node4) at (9,-3) {4};
  \node[] (node6) at (15,-3) {6};
  \node[] (node2) at (0,2) {2};
  \node[line width=1pt] (node1) at (6,2) {1};
  \node[] (node3) at (12,2) {3};
  \draw[latex-latex,very thin] (node2)--(node1);
  \draw[latex-latex,very thin] (node1)--(node3);
  \draw[latex-latex,very thin] (node1)--(node4);
  \draw[latex-latex,very thin] (node3)--(node4);
  \draw[latex-latex,very thin] (node4)--(node5);
  \draw[latex-latex,very thin] (node6)--(node4);

  \draw[-latex,red] (node2).. controls(3.5,1) ..(node4);
  \draw[-latex,red] (node4).. controls(5.5,-2) ..(node2);
  \draw[-latex,red] (node1).. controls(5.5,-0.5) ..(node5);
  \draw[-latex,red] (node5).. controls(3.5,-0.5) ..(node1);
  \draw[-latex,red] (node3).. controls(9,3) ..(node1);
  \draw[-latex,red] (node6).. controls(14,-0.5) ..(node3);

  \draw[dashed,color=black] (2.1,-4.5) -- (11.7,-4.5) -- (6.9,3.5) -- (-2.7,3.5) -- (2.1,-4.5);

\end{tikzpicture}

\end{minipage}
}
\subfigure[]{
\centering
\begin{minipage}[t]{0.45\linewidth}
\centering
\begin{tikzpicture}[scale=0.22, sibling distance=0em,
  every node/.style = {scale=0.65, shape=circle, draw, align=center},
    outline/.style={draw=#1,thick,fill=#1!100}]
  every draw/.style = {scale=2}
  \node[] (node5) at (3,-3) {5};
  \node[] (node4) at (9,-3) {4};
  \node[] (node6) at (15,-3) {6};
  \node[] (node2) at (0,2) {2};
  \node[line width=1pt] (node1) at (6,2) {1};
  \node[] (node3) at (12,2) {3};
  \draw[latex-latex,very thin] (node2)--(node1);
  \draw[latex-latex,very thin] (node1)--(node3);
  \draw[latex-latex,very thin] (node1)--(node4);
  \draw[latex-latex,very thin] (node3)--(node4);
  \draw[latex-latex,very thin] (node4)--(node5);
  \draw[latex-latex,very thin] (node6)--(node4);
  \draw[-latex,red,dashed] (node2).. controls(3,3) ..(node1);
  \draw[-latex,red] (node4).. controls(5.5,-2) ..(node2);
  \draw[-latex,red] (node1).. controls(5.5,-0.5) ..(node5);
  \draw[-latex,red] (node5).. controls(3.5,-0.5) ..(node1);
  \draw[-latex,red] (node3).. controls(9,3) ..(node1);
  \draw[-latex,red] (node6).. controls(14,-0.5) ..(node3);

  \draw[dashed,color=black] (2.1,-4.5) -- (11.7,-4.5) -- (6.9,3.5) -- (-2.7,3.5) -- (2.1,-4.5);

\end{tikzpicture}

\end{minipage}
}
\medskip
\subfigure[]{
\centering
\begin{minipage}[t]{0.45\linewidth}
\centering
\begin{tikzpicture}[scale=0.2, sibling distance=0em,
  every node/.style = {scale=0.65, shape=circle, draw, align=center},
    outline/.style={draw=#1,thick,fill=#1!100}]
  every draw/.style = {scale=2}
  \node[] (node5) at (3,-3) {5};
  \node[] (node4) at (9,-3) {4};
  \node[] (node6) at (15,-3) {6};
  \node[] (node2) at (0,2) {2};
  \node[line width=1pt] (node1) at (6,2) {1};
  \node[] (node3) at (12,2) {3};
  \draw[latex-latex,very thin] (node2)--(node1);
  \draw[latex-latex,very thin] (node1)--(node3);
  \draw[latex-latex,very thin] (node1)--(node4);
  \draw[latex-latex,very thin] (node3)--(node4);
  \draw[latex-latex,very thin] (node4)--(node5);
  \draw[latex-latex,very thin] (node6)--(node4);
  \draw[-latex,red,dashed] (node2).. controls(3,3) ..(node1);
  \draw[-latex,red] (node4).. controls(5.5,-2) ..(node2);
  \draw[-latex,red] (node1).. controls(5.5,-0.5) ..(node5);
  \draw[-latex,red,dashed] (node5).. controls(7.5,-0.5) ..(node3);
  \draw[-latex,red] (node3).. controls(9,3) ..(node1);
  \draw[-latex,red] (node6).. controls(14,-0.5) ..(node3);

  \draw[dashed,color=black] (0.3,-4.5) -- (11.7,-4.5) -- (14.7,3.5) -- (-2.7,3.5) -- (0.3,-4.5);
\end{tikzpicture}

\end{minipage}
}
\subfigure[]{
\centering
\begin{minipage}[t]{0.45\linewidth}
\centering
\begin{tikzpicture}[scale=0.2, sibling distance=0em,
  every node/.style = {scale=0.65, shape=circle, draw, align=center},
    outline/.style={draw=#1,thick,fill=#1!100}]
  every draw/.style = {scale=2}
  \node[] (node5) at (3,-3) {5};
  \node[] (node4) at (9,-3) {4};
  \node[] (node6) at (15,-3) {6};
  \node[] (node2) at (0,2) {2};
  \node[line width=1pt] (node1) at (6,2) {1};
  \node[] (node3) at (12,2) {3};
  \draw[latex-latex,very thin] (node2)--(node1);
  \draw[latex-latex,very thin] (node1)--(node3);
  \draw[latex-latex,very thin] (node1)--(node4);
  \draw[latex-latex,very thin] (node3)--(node4);
  \draw[latex-latex,very thin] (node4)--(node5);
  \draw[latex-latex,very thin] (node6)--(node4);
  \draw[-latex,red,dashed] (node2).. controls(3,3) ..(node1);
  \draw[-latex,red] (node4).. controls(5.5,-2) ..(node2);
  \draw[-latex,red] (node1).. controls(5.5,-0.5) ..(node5);
  \draw[-latex,red,dashed] (node5).. controls(6,-4.3) ..(node4);
  \draw[-latex,red] (node3).. controls(9,3) ..(node1);
  \draw[-latex,red] (node6).. controls(14,-0.5) ..(node3);

  \draw[dashed,color=black] (2.1,-4.5) -- (11.7,-4.5) -- (6.9,3.5) -- (-2.7,3.5) -- (2.1,-4.5);

\end{tikzpicture}

\end{minipage}
}
\medskip
\subfigure[]{
\centering
\begin{minipage}[t]{0.45\linewidth}
\centering
\begin{tikzpicture}[scale=0.2, sibling distance=0em,
  every node/.style = {scale=0.65, shape=circle, draw, align=center},
    outline/.style={draw=#1,thick,fill=#1!100}]
  every draw/.style = {scale=2}
  \node[fill=black!10] (node5) at (3,-3) {5};
  \node[fill=black!10] (node4) at (9,-3) {4};
  \node[] (node6) at (15,-3) {6};
  \node[fill=black!10] (node2) at (0,2) {2};
  \node[fill=black!10] (node1) at (6,2) {1};
  \node[line width=1pt] (node3) at (12,2) {3};
  \draw[latex-latex,very thin] (node2)--(node1);
  \draw[latex-latex,very thin] (node1)--(node3);
  \draw[latex-latex,very thin] (node1)--(node4);
  \draw[latex-latex,very thin] (node3)--(node4);
  \draw[latex-latex,very thin] (node4)--(node5);
  \draw[latex-latex,very thin] (node6)--(node4);
  \draw[-latex,red,dashed] (node2).. controls(3,3) ..(node1);
  \draw[-latex,red] (node4).. controls(5.5,-2) ..(node2);
  \draw[-latex,red] (node1).. controls(5.5,-0.5) ..(node5);
  \draw[-latex,red,dashed] (node5).. controls(6,-4.3) ..(node4);
  \draw[-latex,red,dashed] (node3).. controls(12,-0.5) ..(node6);
  \draw[-latex,red] (node6).. controls(14,-0.5) ..(node3);

\end{tikzpicture}

\end{minipage}
}
\medskip
\subfigure[]{
\centering
\begin{minipage}[t]{0.45\linewidth}
\centering
\begin{tikzpicture}[scale=0.2, sibling distance=0em,
  every node/.style = {scale=0.65, shape=circle, draw, align=center},
    outline/.style={draw=#1,thick,fill=#1!100}]
  every draw/.style = {scale=2}
  \node[fill=black!10] (node5) at (3,-3) {5};
  \node[fill=black!10] (node4) at (9,-3) {4};
  \node[] (node6) at (15,-3) {6};
  \node[fill=black!10] (node2) at (0,2) {2};
  \node[fill=black!10] (node1) at (6,2) {1};
  \node[line width=1pt] (node3) at (12,2) {3};
  \draw[latex-latex,very thin] (node2)--(node1);
  \draw[latex-latex,very thin] (node1)--(node3);
  \draw[latex-latex,very thin] (node1)--(node4);
  \draw[latex-latex,very thin] (node3)--(node4);
  \draw[latex-latex,very thin] (node4)--(node5);
  \draw[latex-latex,very thin] (node6)--(node4);
  \draw[-latex,red,dashed] (node2).. controls(3,3) ..(node1);
  \draw[-latex,red] (node4).. controls(5.5,-2) ..(node2);
  \draw[-latex,red] (node1).. controls(5.5,-0.5) ..(node5);
  \draw[-latex,red,dashed] (node5).. controls(6,-4.3) ..(node4);
  \draw[-latex,red,dashed] (node3).. controls(12,-0.5) ..(node6);
  \draw[-latex,red,dashed] (node6).. controls(14,-0.5) and (16,-0.5) ..(node6);
  \draw[dashed,color=black] (13,-4.5) --(17,-4.5)-- (17,-0.5)--(13,-0.5)-- (13,-4.5);

\end{tikzpicture}

\end{minipage}
}
\subfigure[]{
\centering
\begin{minipage}[t]{0.45\linewidth}
\centering
\begin{tikzpicture}[scale=0.2, sibling distance=0em,
  every node/.style = {scale=0.65, shape=circle, draw, align=center},
    outline/.style={draw=#1,thick,fill=#1!100}]
  every draw/.style = {scale=2}
  \node[fill=black!10] (node5) at (3,-3) {5};
  \node[fill=black!10] (node4) at (9,-3) {4};
  \node[fill=black!10] (node6) at (15,-3) {6};
  \node[fill=black!10] (node2) at (0,2) {2};
  \node[fill=black!10] (node1) at (6,2) {1};
  \node[line width=1pt] (node3) at (12,2) {3};
  \draw[latex-latex,very thin] (node2)--(node1);
  \draw[latex-latex,very thin] (node1)--(node3);
  \draw[latex-latex,very thin] (node1)--(node4);
  \draw[latex-latex,very thin] (node3)--(node4);
  \draw[latex-latex,very thin] (node4)--(node5);
  \draw[latex-latex,very thin] (node6)--(node4);
  \draw[-latex,red,dashed] (node2).. controls(3,3) ..(node1);
  \draw[-latex,red] (node4).. controls(5.5,-2) ..(node2);
  \draw[-latex,red] (node1).. controls(5.5,-0.5) ..(node5);
  \draw[-latex,red,dashed] (node5).. controls(6,-4.3) ..(node4);
  \draw[-latex,red,dashed] (node3).. controls(11,4) and (13,4) ..(node3);
  \draw[-latex,red,dashed] (node6).. controls(14,-0.5) and (16,-0.5) ..(node6);
  
  \draw[dashed,color=black] (10,0) --(14,0)-- (14,4)--(10,4)-- (10,0);
\end{tikzpicture}

\end{minipage}
}

\caption{An running example of CTC. The double arrows represent agents' connections. The red arrow is agents' pointing in the favorite pointing graph. The dashed red arrow means the agent switch her pointing to a less favorite agent. We use a dashed box to indicate agents in the minimum connected component $T$, and the matched agents are in grey.}
\label{fig:run_example}
\Description{Execution of \textit{Connected Trading Cycles} on Figure~\ref{fig:pref_table}.}
\end{figure}
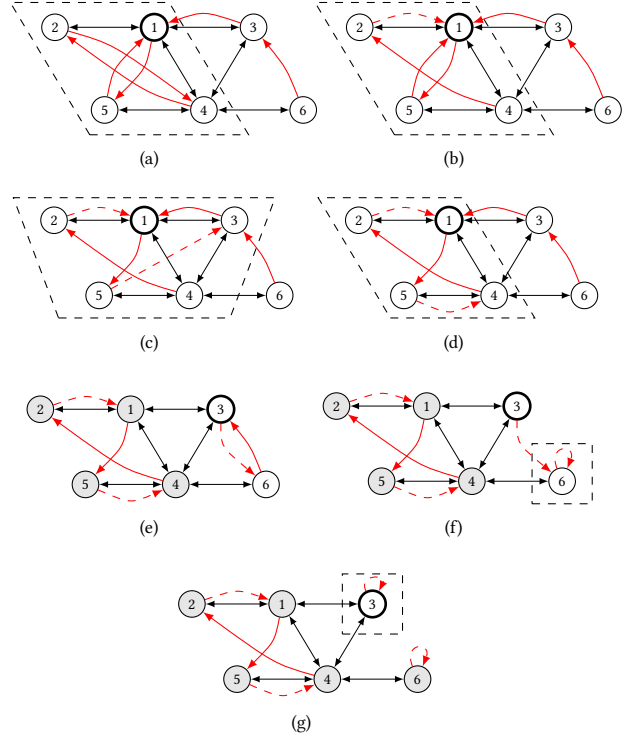

In this section, We run a detailed example to illustrate our mechanism. The initial agent set is $N_0=\{ 1 \}$. The order is $\mathcal{O}=(1,2,3,4,5,6)$. The social network, type profiles, and allocation are presented in Figure~\ref{fig:pref_table}. The following steps match the sub-figures in Figure~\ref{fig:run_example}. Since edges in both $G(\theta')$ and $F(\theta')$ are directed, which may cause ambiguity, we adopt $i \to j$ to represent the directed edge $\langle i,j\rangle$ in $G(\theta')$ and use $\langle i,j \rangle$ for directed edge between $i,j$ in $F(\theta')$. For simplicity, we denote a path $(\langle i,j \rangle,\langle j,k \rangle)$ in $G(\theta')$ as $(i \to j \to k)$.

\begin{enumerate}
    \item[(a)] In the beginning, agent 1 has the minimum order in $N \setminus V=\{ 1,2,3,4,5,6 \}$. We detect $P_5 = C_5 = \{1,5\}$. In the path detection process, the minimum connected component that contains $C_5$ is $T = \{1,2,4,5\}$. We then compute the shortest path set $SP=\{(\textbf{1}\to4\to5),(\textbf{2}\to1\to4),(\textbf{4}\to1\to2),(\textbf{5}\to4\to1)\}$. Both agent 2 and agent 5 do not have a path to their pointing only marked by themselves, and their outgoing edges are not marked by others, so $S=\{2,5\}$. Since $(T\setminus C_5)\cap S = \{2\}$, and agent 2 has the minimum order in $S$, we change $\langle 2,4\rangle$ to $\langle 2,\succ_2^{next}(F(\theta'))=3\rangle$ (CTC step c-iv).
    
    \item[(b)] Agent 1 has the minimum order in $N \setminus V=\{ 1,2,3,4,5,6\}$. We detect $P_5 = C_5 = \{1,5\}$. In the path detection process, the minimum connected component that contains $C_5$ is $T = \{1,2,4,5\}$. The shortest path set is $SP=\{(\textbf{2}\to1),(\textbf{1}\to4\to5),(\textbf{4}\to1\to2),(\textbf{5}\to4\to1)$. Agent 5 is the only agent who does not have a path to her pointing only marked by herself, and her outgoing edge is not marked by others, so $S=\{5\}$. Since $(T\setminus C_m) \cap S = \emptyset$ and $S=\{5\}$, we change $\langle5,1\rangle$ to $\langle 5,\succ_5^{next}(F(\theta'))=3\rangle$ (CTC step c-v). 

    \item[(c)] Agent 1 has the minimum order in $N \setminus V=\{ 1,2,3,4,5,6\}$. We detect $P_3 = C_3 = \{1,5,3\}$. In the path detection process, the minimum connected component that contains $C_3$ is $T=\{1,2,3,4,5\}$. The shortest path set $SP=\{(\textbf{2}\to1),(\textbf{3}\to1),(\textbf{1}\to4\to 5),(\textbf{4}\to1\to2),(\textbf{5}\to4\to3)\}$. All agents in $T$ have a path to their pointing only marked by themselves, so $S=\emptyset$ and $C_3\subset T$. Agent 2 has the minimum order in $T\setminus C_3=\{2,4\}$, we start from agent 2's pointing, and agent 5 is the last agent on $C_3$ whose path to her pointing contains at least one agent (i.e., agent 4) in $T\setminus C_3$. We change $\langle 5,3 \rangle$ to $\langle 5,\succ_5^{next}(F(\theta'))=4\rangle$ (CTC step c-iii).
    
    \item[(d)] Agent 1 has the minimum order in $N \setminus V=\{ 1,2,3,4,5,6 \}$. We detect $P_2 = C_2 = \{1,5,4,2\}$. In the path detection process, the minimum connected component that contains $C_2$ is $T = \{1,2,4,5\}$. We then compute the shortest path set $SP=\{(\textbf{2}\to1),(\textbf{5}\to4),(\textbf{1}\to4\to5),(\textbf{4}\to1\to2)$. All agents in $T$ have a path to their pointing only marked by themselves, so $S=\emptyset$ and $C_2=T$. We add $\{1,2,4,5\}$ to $V$. Since agent 3 belongs to $N\setminus V$ but she points to agent 1 who is added to $V$, we change $\langle 3,1 \rangle$ to $\langle 3,\succ_3^{next}(F(\theta'))=6\rangle$ (CTC step c-ii).
    
    \item[(e)] Agent 3 has the minimum order in $N \setminus V=\{ 3,6 \}$. We detect $P_6 =C_6=\{3,6\}$. In the path detection process, we cannot find a connected component $T$ that contains $C_6$. Agent 6 is the last agent on $C_6$ who cannot connect to her pointing, we change $\langle 6,3 \rangle$ to $\langle 6,\succ_6^{next}(F(\theta'))=6\rangle$ (CTC step c-i). 
    
    \item[(f)] Agent 3 has the minimum order in $N \setminus V=\{3,6\}$. We detect $P_6 = \{3,6\}$ and $C_6=\{6\}$. In the path detection process, the minimum connected component that contains $C_6$ is $T=C_6=\{6\}$. The shortest path set $SP=\{(\textbf{6}\to6)\}$. Agent 6 is the only agent in $T$, and she has a path to herself that is not marked by others. We have $S=\emptyset$. So, we add $\{6\}$ to $V$ (CTC step c-ii). Since agent 3 belongs to $N\setminus V$ but she points to agent 6 who is added to $V$, we change $\langle 3,6 \rangle$ to $\langle 3,\succ_3^{next}(F(\theta'))=3\rangle$ (CTC step c-ii).

    \item[(g)] Agent 3 has the minimum order in $N \setminus V=\{3\}$. We detect $P_3 = C_3 = \{3\}$. The shortest path set $SP=\{(\textbf{3}\to3)\}$. Agent 3 is the only agent in $T$, and she has a path to herself that is not marked by others. We have $S=\emptyset$. So, we add $\{3\}$ to $V$ (CTC step c-ii).
    
    \item[(h)] Now $V=N$, the algorithm terminates. The allocation given by CTC is $\pi=(h_5,h_1,h_3,h_2,h_4,h_6)$.
\end{enumerate}

\section{Formal Proofs for Mechanisms}
\subsection{Swap With Neighbors}
In this section, we prove that SWN is IR, IC and stable-cc.

\begin{theorem}
SWN is IR.
\end{theorem}

\begin{proof}
 In SWN, agent $i$ can always choose herself as her favorite agent, and then SWN will allocate $h_i$ to $i$. Thus, SWN is IR.   
\end{proof}

\begin{theorem}
SWN is IC.
\end{theorem}

\begin{proof}
Since each agent $i$'s type consists of two parts, her preference $\succ_i$ and her neighbor set $r_i$, we will prove misreporting neither $\succ_i$ nor $r_i$ can improve her allocation. 

\noindent \textbf{Misreport on $\succ$:} For agent $i$, we fix her reported neighbor set as $r_i'$. Her real preference is $\succ_i$ and her reported preference is $\succ_i'$. Suppose her allocation $\pi_i(( \succ_i', r_i' ), \theta'_{-i}) = h_{j'} \succ_i \pi_i(( \succ_i, r_i' ), \theta'_{-i}) = h_j$. 

In SWN, when $i$ truthfully reports her preference, agent $i$ is allocated $h_j$ instead of $h_{j'}$, we know that $h_{j'}$ is in a trading cycle without $i$. Let the cycle containing $j'$ be $C_{j'}$. When we fix the others preference and $i$ misreports as $(\succ'_i,r_i')$, since the trading cycle $C_{j'}$ is only determined by agents' preference in $C_{j'}$, which excludes agent $i$, $C_{j'}$ still forms. By SWN, agent $i$ cannot be allocated $h_{j'}$, which contradicts our assumption.

\noindent \textbf{Misreport on $r$:} For agent $i$, we fix her reported preference as $\succ_i'$. Her real neighbor set is $r_i$ and her reported neighbor set is $r_i'$. Now we suppose her allocation $\pi_i(( \succ_i', r_i' ), \theta'_{-i}) = h_{j'} \succ_i' \pi_i(( \succ_i', r_i ), \theta'_{-i}) = h_j$. 

In SWN, each agent can only be allocated the item from her reported neighbor set $r_i'$. Therefore, we have $j \in r_i$, $j' \in r_i$ and $j' \in r_i'$. Since with true neighbor set $r_i$, agent $i$ is allocated $h_j$ instead of $h_{j'}$, we know that $h_{j'}$ is in a trading cycle without $i$. Let the cycle containing $j'$ be $C_{j'}$. According to SWN, each agent in $C_{j'}$ is pointing to her neighbor, which means their neighbor set is irrelevant to $r_i$ as well as any $r_i'$. Therefore, the trading cycle $C_{j'}$ still remains and excludes agent $i$ when agent $i$ misreports $r_i'$. By SWN, agent $i$ cannot be allocated $h_{j'}$, which contradicts our assumption.

Put the above two steps together, we prove that SWN is IC. 
\end{proof}

\begin{theorem}
SWN is stable-cc.
\end{theorem}

\begin{proof}
For every $S\subseteq N$ and their item set $H_S$. Let the allocation given by SWN be $\pi(\theta)$. If there exists a blocking coalition $S$, where $S \subseteq N$ is the node set of a complete component in $G(\theta)$. 

Since $S$ is the node set of a complete component, we have $\forall i \in S, S \subseteq r_i$. A blocking coalition $S$ suggests there exists a $z(\theta)$ such that for all  $i\in S$, $z_i(\theta) \in H_S$,  $z_i(\theta) \succeq_i \pi_i(\theta)$ with at least one $j\in S$ we have $z_j(\theta) \succ_j \pi_j(\theta)$. Therefore, for all $ j \in S$, the blocking coalition guarantees the owner of $z_j(\theta)$ and $j$ are neighbors. Based on SWN, $z_j(\theta) \succ_j \pi_j(\theta)$ means $j$ can always point to and get allocated $z_j(\theta)$ instead of $\pi_j(\theta)$. Thus, the trading cycle which contains the owner of $z_j(\theta)$ and $j$ can trade by the cycle (i.e.,$\forall i \in S, z_i(\theta)$ = $\pi_i(\theta)$). This contradicts the assumption of existing at least one $j \in S \subseteq r_j, z_j(\theta) \succ_j \pi_j(\theta)$. Hence, SWN is stable-cc.
\end{proof}

\subsection{Leave and Share}
In this section, we prove that LS is IR, IC and stable-cc.

\begin{theorem}
For any order $\mathcal{O}$, LS is IR.
\end{theorem}

\begin{proof}
In LS, agent $i$ leaves only when she gets an item $h_j$. Agent $i$ can always choose herself as her favorite agent, and then LS will allocate $h_i$ to $i$. Thus, LS is IR.
\end{proof}

\begin{theorem}
For any order $\mathcal{O}$, LS is IC.
\end{theorem}

\begin{proof}
Since each agent $i$'s type consists of two parts, her preference $\succ_i$ and her neighbor set $r_i$, we will prove misreporting neither $\succ_i$ nor $r_i$ can improve her allocation. 

\noindent \textbf{Misreport on $\succ$:} For agent $i$, we fix her reported neighbor set as $r_i'$. Her real preference is $\succ_i$ and her reported preference is $\succ_i'$. Now we compare her allocation $\pi_i(( \succ_i, r_i' ), \theta_{-i}) = h_j$ with $\pi_i(( \succ_i', r_i' ), \theta_{-i}) = h_{j'}$. 

Since $\mathcal{O}$ is based on the minimum distance, which is irrelevant to agents' preferences, we only need to prove that $h_j \succeq_i h_{j'}$ for all agents for a given order.

Before $i$ is pushed into the stack, all trading cycles are irrelevant to $\succ_i'$( $i$ has not been preferred by the agents in the stack before, so $\succ_i'$ is not used at all). Thus we only consider the situation when agent $i$ is pushed into the stack and then $\succ_i'$ can decide which agent after $i$ is pushed into the stack. 


When $i$ is on the top of the stack, the next pushed agent $f_i(R_i)$ is determined by $\succ_i'$. Agent $i$ can be allocated with $h_{j'}$ only when there is a trading cycle with $i$. Assume that $h_{j'} \succ_i h_j$, i.e., misreporting $\succ_i$ gives $i$ a better item. We will show this leads to a contradiction. If $i$ reported $\succ_i$ truthfully, then $i$ would first choose $j'$ before $j$ ($j'$ is pushed into the stack first), since $i$ did not get $h_{j'}$, which means $j'$ formed a cycle $C_{j'}$ without $i$. If reporting $\succ_i'$, $i$ is matched with $j'$, then it must be the case that there exists another trading cycle $B$ which breaks the cycle $C_{j'}$. Otherwise, whenever $i$ points to $j'$, $j'$ will form the original cycle $C_{j'}$ as it is independent of $i$'s preference. The only possibility for $i$ to achieve this is by pointing her favorite agent under the false preference $\succ_i'$. By doing so, $i$ can force other agents to leave earlier with different cycles including $B$. Next, we will show that it is impossible for $B$ to break $C_{j'}$.

If $B$ can actually break $C_{j'}$, there must be an overlap between $B$ and $C_{j'}$. Assume that $x$ is the node where $B$ joins $C_{j'}$ and $y$ is the node where $B$ leaves $C_{j'}$ ($x$ and $y$ can be the same node). For node $y$, her match in $B$ and $C_{j'}$ cannot be the same (the model assumes strict preference), and no matter when $y$ is pushed into the stack, both items in $B$ and $C_{j'}$ are still there. Assuming the matching in $C_{j'}$ is her favorite, then cycle $B$ will never be formed. This contradicts to $h_{j'} \succ_i h_j$, so reporting $\succ_i$ truthfully is a dominant strategy.

\noindent \textbf{Misreport on $r$:} As the above showed for any reported neighbor set $r_i'$, reporting $\succ_i$ truthfully is a dominant strategy. Next, we further show that under truthful preference reports, reporting $r_i$ is a dominant strategy.
That is, for the allocation $\pi_i (( \succ_i, r_i ), \theta_{-i}) = h_j$ and $\pi_i(( \succ_i, r_i' ), \theta_{-i}) = h_{j'}$, we will show $h_j \succeq_i h_{j'}$.

Firstly, we show that the tradings before $i$ being pushed into the stack are irrelevant to $i$'s neighbor set report $r_i'$. 
For all the agents ranked before $i$ in $\mathcal{O}$, their shortest distance is smaller than or equal to $i$'s shortest distance to agent $1$, which means that their shortest paths do not contain $i$ and therefore $r_i'$ cannot change them. Thus, $r_i'$ cannot change the order of all agents ordered before $i$ in $\mathcal{O}$. In addition, agent $i$ could be a cut point to disconnect certain agents $D_i$ from agent $1$, so $r_i'$ can impact $D_i$'s distances and qualification. However, $D_i$ can only be involved in the matching after $i$ is in the stack, as others cannot reach $D_i$ without $i$. 
Hence, before $i$ is pushed into the stack, the tradings only depend on those ordered before $i$ and the agents excluding $D_i$, which are independent of $i$. In fact, the order of the agents pushed into the stack before $i$ is the same no matter what $r_i'$ is. That is, when $i$ is pushed into the stack, the agents, except for $D_i$, remaining in the game are independent of $i$.

Then when $i$ misreports $r_i$, she will only reduce her own options in the favorite agent selection. Whether $r_i'$ disconnects $D_i$ or not, reporting $r_i'$ here is equivalent to modifying $\succ_i$ by disliking neighbors in $r_i\setminus r_i'$. As we have shown, this is not beneficial for the agent. Therefore, reporting $r_i$ truthfully is a dominant strategy, i.e., $h_j \succeq_i h_{j'}$.

Put the above two steps together, we have proved that $\pi_i(( \succ_i, r_i ), \theta_{-i}) \succeq_i \pi_i(( \succ_i', r_i' ), \theta_{-i})$, i.e., LS is IC. 
\end{proof}

\begin{theorem}
For any order $\mathcal{O}$, LS is stable-cc.
\end{theorem}

\begin{proof}
For every $S\subseteq N$ and their item set $H_S$. Let the allocation given by LS be $\pi(\theta)$. If there exists a blocking coalition $S$, where $S \subseteq N$ is the node set of a complete component in $G(\theta)$, we have $\forall i \in S, S \subseteq r_i$. A blocking coalition $S$ suggests there exists a $z(\theta)$ such that for all  $i\in S$, $z_i(\theta) \in H_S$,  $z_i(\theta) \succeq_i \pi_i(\theta)$ with at least one $j\in S$ we have $z_j(\theta) \succ_j \pi_j(\theta)$. Therefore, for all $ j \in S$, the blocking coalition guarantees the owner of $z_j(\theta)$ and $j$ are in one trading cycle. This indicates if a trading cycle contains any agent in the coalition, all the agents in the trading cycle are in the coalition. Based on LS, $z_j(\theta) \succeq \pi_j(\theta)$ means the owner of $z_j(\theta)$ will be pushed into the stack before the owner of $\pi_j(\theta)$. Thus, the trading cycle which contains the owner of $z_j(\theta)$ and $j$ can trade by the cycle (i.e.,$\forall i \in S, z_i(\theta)$ = $\pi_i(\theta)$). This contradicts the assumption of existing $j \in S, z_j(\theta) \succ_j \pi_j(\theta)$. Hence, LS is stable-cc.
\end{proof}

\subsection{Connected Trading Cycles}
In this section, we prove that CTC is IR, IC, and stable-cc. The proof for optimal-cc is in Section 5, Theorem 5.7.

\begin{theorem}
For any order $\mathcal{O}$, CTC is IR.
\end{theorem}
\begin{proof}
In CTC, every agent $i$ can point to her endowment $h_i$, and $i$ always has an exclusive path to herself so she does not have to switch her preference. Eventually, CTC will assign $h_i$ to $i$. Thus, CTC is IR.
\end{proof}

\begin{theorem} 
For any order $\mathcal{O}$, CTC is IC.
\end{theorem}

\begin{proof}
    In CTC, the path detection and preference switching process (in CTC step c) ensures each agent choose from her best possible choice (based on the next favorite function), and only switch pointing if all paths from herself to her pointing fail to be connected by an exclusive path. In path detection, while agent $i$'s path to her pointing $j$ overlaps with others, we constantly find all possible paths from $i$ to $j$ with ascending order in length. This means, $i$ cannot misreport neighbor set ($r'_i\subset r_i$) to add new paths to $j$. Further, $i$ has no incentive to misreport neighbor set to exclude her competitor $i'$ (i.e., $i$ can misreport $r'_i\subset r_i$ so that $i'$ who also points to $j$ is unqualified.) because in this case, all paths from $i'$ to $j$ pass through $i$. So, if $i$ and $i'$ are competing for the same item $h_j$, $i'\to j$ overlaps with $i\to j$, so $i'$ will switch to a less preferred agents before $i$ does.
    
    Suppose CTC is not IC, there must exist an agent $i\in N$, when $i$ truthfully report her type, we have $\pi_i((\theta_i, \theta'_{-i}))=h_j$, and $i$ can misreport to improve her matching, $\pi_i((\theta'_i, \theta'_{-i}))=h_{j'}\succ_i h_j$.

    We denote the cycle involving $i\to j$ as $C_i$, suppose $\pi_k((\theta_i,\theta'_{-i}))=h_i$ and $\pi_i((\theta_i, \theta'_{-i}))=h_j$. According to CTC step c.ii, we know that cycle $C_i$ is a connected cycle and each agent in $C_i$ has an exclusive path to her pointing.

    Similarly, we denote the cycle formed under $i$'s misreport as $C'_i$, involving $i\to j'$. Suppose we have $\pi_{k'}((\theta'_i,\theta'_{-i}))=h_i$ and $\pi_i((\theta'_i, \theta'_{-i}))=h_{j'}$. According to CTC step c.ii, we know that cycle $C'_i$ is also a connected cycle and each agent in $C'_i$ has an exclusive path to her pointing.

    Since other agents type profiles are fixed as $\theta'_{-i}$, $j\to \cdots \to k \to i$ in cycle $C_i$ and $j'\to \cdots \to k' \to i$ in cycle $C'_i$ form irrelevant of agent $i$'s reported type profile. Thus, following truthful type $\theta_i$, agent $i$ will point to $j'$ first ($h_{j'}\succ_i h_j$), so $C'_i$ will form and get traded in CTC.

    This contradicts our assumption which means CTC is IC.
\end{proof}

\begin{theorem}
For any order $\mathcal{O}$, CTC is stable-cc.
\label{Stable-cc}
\end{theorem}
\begin{proof}
For any given type profile $\theta$, if the allocation $\pi(\theta)$ given by CTC violates stable-cc, there exists a group of agents $B$ that forms a complete component in $G(\theta)$, they deviate together and swap among themselves can result in a better allocation $\pi'(\theta)$. Since group $B$ is a complete component, each agent $i$ can point to the owner of her allocation $\pi'_i(\theta)$ through her outgoing edge (i.e., $\pi'_i(\theta)=h_j$, agent $i$ points to $j\in B$). 

We next demonstrate that when agent $i\in B$ points to $\pi_i'(\theta)$, $i$ will not switch her pointing to a less preferred agent in CTC. Suppose during CTC's execution, an agent $i\in B$ points to $j\in B, \pi_i'(\theta)=h_j$, let the connected component involving $i\to j$ be $T$. Since $i\to j$ is $i$'s outgoing edge, whether $i\to j$ is exclusively marked by $i$ or not, $i$ is not in $S$ (refer to CTC step c). If $T=\emptyset$ (CTC step c.i), CTC will not let $i$ switch her pointing because $i$ can always connect to her pointing $j$ by her outgoing edge. If $S=\emptyset$ and $C_m=T$ (CTC step c.ii), CTC will let $C_m$ trade, so $i$ still does not switch points. If $S=\emptyset$ and $C_m\subset T$ (CTC step c.iii), since $i$'s pointing path passes no other agents, $i$ does not have to switch points. Finally, if $S\neq \emptyset$ (CTC step c.iv and c.v), since $i\notin S$, CTC will not let $i$ switch her pointing.

That is to say, any agent $i\in B$ can hold her pointing to $\pi'_i(\theta)$, and CTC will allocate $\pi'(\theta)$. This means $\pi(\theta)=\pi'(\theta)$, which contradicts our assumption. 
\end{proof}


\subsection{Complexity Analysis}
\begin{proposition}
    The computation complexity of SWN is $O(N^2)$.
\end{proposition}
\begin{proof}
    In SWN, each agent points to the neighbor with her favorite item, and then we detect cycles in the constructed directed pointing graph with the computation cost of $O(N)$. In the worst-case scenario, all agents are self-looped, the above process will be conducted $N$ times, so the computation complexity is $O(N^2)$.
\end{proof}

\begin{proposition}
    The computation complexity of LS is $O(N^2)$.
\end{proposition}
\begin{proof}
    In LS, we first compute the shortest path length from each agent to the initial agent set to determine the order $\mathcal{O}$. The corresponding computation complexity is $O(N^2)$. Next, following the order, at each turn, we start from the agent with the smallest order to form a cycle with a computation complexity of $O(N)$. A stack is used to record the pointed agents, if a sequence of agents forms a cycle, they will be popped out of the stack (as in LS step b). Once the stack is empty, LS merges the remaining neighbors of the matched agents with a computation complexity of $O(N)$ (as in LS step c). In the worst-case scenario, all agents are self-looped, so the above process will be conducted $N$ times, and the total computation complexity is $O(N^2)$.
\end{proof}

\begin{proposition}
    The computation complexity of CTC is $O(2^N\cdot N^2)$.
\end{proposition}
\begin{proof}
    In CTC, we first compute the order for each agent with a computation complexity of $O(N^2)$. Then, starting from the agent with the smallest order, CTC finds a cycle sequence $C_m$ in the favorite pointing graph with computation complexity of $O(N)$ (as in CTC step a). During the path detection process, the construction of minimum connected component $T$ requires a combination of all agents in $N\setminus C_m$ and a check on the connectedness for each combination attempt with computation complexity of $O(N^2)$, resulting in a computation cost of $O(2^N\cdot N^2)$ (as in Path Detection step a). Next, CTC applies a marking process to determine whether each agent in $T$ has an exclusive path to her pointing to further distinguish the connectedness of $C_m$ (as in Path Detection step c,d). This process involves determining the shortest paths for all agents in $T$, sorting the paths, and marking on a subgraph $G_T$. Together, the computation cost for the marking process is $O(N^2\cdot Nlog N)$. Finally, with the marked $G_T$, CTC determines which agent should switch her pointing to the next favorite or let the cycle $C_m$ get traded in $O(N)$ as in CTC step c. Hence, the computation complexity of CTC is $O(2^N\cdot N^2)$.
\end{proof}


\end{document}